\documentclass[a4paper]{article}

\usepackage{a4wide}
\usepackage[english]{babel}
\usepackage{latexsym}
\usepackage{amssymb}
\usepackage{amsmath}
\usepackage{amsthm}
\usepackage{cite}    
\usepackage{ifthen}
\usepackage{xspace}
\usepackage{tikz}
\usepackage{paralist}

\usepackage[colorlinks,final,citecolor=blue]{hyperref}

\usetikzlibrary{%
  arrows,%
  positioning,%
  decorations.pathmorphing,%
  decorations.pathreplacing,%
} 

\theoremstyle{plain}
\newtheorem{theorem}{Theorem}[section]

\newtheorem{lemma}[theorem]{Lemma}
\newtheorem{corollary}[theorem]{Corollary}

\theoremstyle{definition}

\newtheorem{example}[theorem]{Example}

\theoremstyle{remark}
\newtheorem{remark}[theorem]{Remark}

\newcommand{\cS}{\ensuremath{\mathcal{S}}\xspace}

\renewcommand{\phi}{\varphi}

\newcommand{\alp}{\alpha}
\newcommand{\bet}{\beta}
\newcommand{\gam}{\gamma}
\newcommand{\del}{\delta}

\newcommand{\sig}{\sigma}

\newcommand{\Sig}{\Sigma}

\newcommand{\sse}{\subseteq}
\newcommand{\ssne}{\subsetneq}

\newcommand{\sm}{\setminus}

\newcommand{\ie}{i.\,e.,\xspace}
\newcommand{\eg}{e.\,g.,\xspace}

\newcommand{\resp}{resp.\xspace}

\newcommand\ti[1]{\tilde{#1}}

\newcommand{\set}[2]{\left\{#1\mathrel{\left|\vphantom{#1}\vphantom{#2}\right.}#2\right\}}
\newcommand{\oneset}[1]{\left\{\mathinner{#1}\right\}}

\newcommand{\abs}[1]{\left|\mathinner{#1}\right|}

\newcommand{\N}{\mathbb{N}}

\newcommand{\Oh}{\mathcal{O}}

\newcommand{\e}{\varepsilon}

\newcommand{\PSPACE}{\ensuremath{\mathrm{PSPACE}}\xspace}

\newcommand{\eLL}{\le_{\ell\ell}}
\newcommand{\LL}{<_{\ell\ell}}

\newcommand\pos[1]{{[#1]}}

\newcommand{\sfA}{\ensuremath{\mathsf{A}}\xspace}
\newcommand{\sfC}{\ensuremath{\mathsf{C}}\xspace}
\newcommand{\sfG}{\ensuremath{\mathsf{G}}\xspace}
\newcommand{\sfT}{\ensuremath{\mathsf{T}}\xspace}

\newcommand{\bA}{\sfA}
\newcommand{\bC}{\sfC}
\newcommand{\bG}{\sfG}
\newcommand{\bT}{\sfT}

\newcommand{\smalloverline}[1]
{{\mspace{1.5mu}\overline{\mspace{-1.5mu}#1\mspace{-1.5mu}}\mspace{1.5mu}}}
\newcommand{\ov}[1]{\smalloverline{#1}\vphantom{#1}}

\newcommand\splice[1]{\mathrel{\vdash_{#1}}}

\newcommand\abgd{\alp\bet\gam\del}
\newcommand\abg{\alp\bet\gam}
\newcommand\abjgd{\alp\bet^j\gam\del}
\newcommand\abjg{\alp\bet^j\gam}

\newcommand\I{(i)\xspace}
\newcommand\II{(ii)\xspace}

\newcommand\A{(a)\xspace}
\newcommand\B{(b)\xspace}

\colorlet{eins}{blue!90!black}
\colorlet{zwei}{red!90!black}
\colorlet{drei}{green!70!black}
\colorlet{vier}{orange!90!black}

\newcommand{\pumping}[1]{%
	\begin{enumerate}
		\item $\ti z := z$;
		\item if $\ti z\pos{k;k+#1} = \abg$ for some $k$ such that neither
		\begin{enumerate}[(a)]
			\item $\alp\bet^{j\slash 2}$ is a factor of $\ti z$ starting at position
				$\ti z \pos k$ nor
			\item $\bet^{j\slash 2}\gam$ is a factor of $\ti z$ ending at position
				$\ti z \pos{k+#1}$,
		\end{enumerate}
		then let $\ti z := \ti z\pos{0;k} \cdot\abjg \cdot\ti z\pos{k+#1;\abs{\ti z}}$;
		(replace the factor $\ti z\pos{k;k+#1} = \abg$ in $\ti z$ by $\abjg$)
		\item repeat step 2 until there is no such factor $\abg$ in $\ti z$ left.
	\end{enumerate}%
}

\begin{document}

\title{Deciding Whether a Regular Language is\\
	Generated by a Splicing System}

\author{Lila Kari \and Steffen Kopecki}

\date{}

\maketitle

{\small\centering
	Department of Computer Science \\
	The University of Western Ontario \\
	Middlesex College, London ON N6A 5B7 Canada \\
	{\tt\{lila,steffen\}@csd.uwo.ca} \par
}

\begin{abstract}
Splicing as a binary word/language operation is inspired by the
DNA recombination under the action of restriction enzymes and ligases, and
was first introduced by Tom Head in 1987. Shortly thereafter, it was proven
that the languages generated by (finite) splicing systems form a proper
subclass of the class of regular languages. However, the question of
whether or not one can decide if a given regular language is generated by
a splicing system remained open. In this paper we give a positive answer
to this question. Namely, we prove that, if a language is generated by a
splicing system, then it is also generated by a splicing system whose size
is a function of the size of the syntactic monoid of the input language,
and which can be effectively constructed.
\end{abstract}

\section{Introduction}

\newcommand{\strandA}{
	\node at (.15,-.2) {$\bC$};
	\node at (-.15,-.2) {$\bG$};
	\draw (-.45,-.2) node [fill=white] {$\bA$} --
		++(-.75,0) node [fill=white] {$\ov\alp$} --
		++(-.75,0) node [fill=white] {$3'$};
	\draw (-.45,.2) node [fill=white] {$\bT$} --
		++(-.75,0) node [fill=white] {$\alp$} --
		++(-.75,0) node [fill=white] {$5'$};
}

\newcommand{\strandB}{
	\node at (-.15,.2) {$\bC$};
	\node at (.15,.2) {$\bG$};
	\draw (.45,-.2) node [fill=white] {$\bT$} --
		++(.75,0) node [fill=white] {$\ov\bet$} --
		++(.75,0) node [fill=white] {$5'$};
	\draw (.45,.2) node [fill=white] {$\bA$} --
		++(.75,0) node [fill=white] {$\bet$} --
		++(.75,0) node [fill=white] {$3'$};
	\draw [dotted] (-.3,.35) -- (-.3,0) -- (.3,0) -- (.3,-.35);
}

\newcommand{\strandC}{
	\node at (.15,-.2) {$\bC$};
	\node at (-.15,-.2) {$\bG$};
	\draw (-.45,-.2) node [fill=white] {$\bC$} --
		++(-.75,0) node [fill=white] {$\ov\gam$} --
		++(-.75,0) node [fill=white] {$3'$};
	\draw (-.45,.2) node [fill=white] {$\bG$} --
		++(-.75,0) node [fill=white] {$\gam$} --
		++(-.75,0) node [fill=white] {$5'$};
	\draw [dotted] (-.3,.35) -- (-.3,0) -- (.3,0) -- (.3,-.35);
}

\newcommand{\strandD}{
	\node at (-.15,.2) {$\bC$};
	\node at (.15,.2) {$\bG$};
	\draw (.45,-.2) node [fill=white] {$\bG$} --
		++(.75,0) node [fill=white] {$\ov\del$} --
		++(.75,0) node [fill=white] {$5'$};
	\draw (.45,.2) node [fill=white] {$\bC$} --
		++(.75,0) node [fill=white] {$\del$} --
		++(.75,0) node [fill=white] {$3'$};
}

In \cite{Head87} Head described an language-theoretic operation, called
{\em splicing}, which models DNA
recombination, a cut-and-paste operation on DNA double-strands.
Recall that a {\em DNA single-strand} is a polymer consisting of a series of the nucleobases Adenine (\bA), Cytosine (\bC), Guanine (\bG), and Thymine (\bT) attached to a linear, directed backbone.
Due to the chemical structure of the backbone, the ends of a single-strand are called $3'$-end and $5'$-end.
Abstractly, a DNA single-strand can be viewed as a string over the four letter alphabet $\oneset{\bA,\bC,\bG,\bT}$.
The bases \bA and \bT, respectively \bC and \bG, are {\em Watson-Crick-complementary}, or simply {\em complementary}, which means they can attach
to each other via hydrogen bonds.
The {\em complement} of a DNA single-strand 
$\alp = 5'\text-a_1\cdots a_n\text-3'$ is the strand 
$\ov\alp = 3'\text-\ov{a_1}\cdots \ov{a_n}\text-5'$ where $a_1,\ldots,a_n$ are bases and $\ov{a_1},\ldots,\ov{a_n}$ denote their complementary bases, respectively; note that $\alp$ and $\ov\alp$ have opposite orientation.
A strand $\alp$ and its complement $\ov\alp$ can bond to each other to form a {\em DNA (double-)strand}.

Splicing is meant to abstract the action of two compatible restriction enzymes and the ligase enzyme on two DNA double-strands.
The first restriction enzyme recognizes a base-sequence $u_1 v_1$, called its {\it restriction site}, in any DNA string, and cuts the string containing this factor between $u_1$ and $v_1$.
The second restriction enzyme, with restriction site $u_2 v_2$, acts similarly.  Assuming that the {\em sticky ends} obtained after these cuts are complementary, the enzyme ligase aids then the recombination (catenation) of the first segment of
one cut string with the second segment of another cut string.
For example, the enzyme {\it Taq}\,I has restriction site $\bT\bC\bG\bA$, and the enzyme {\it Sci}\,NI has restriction site $\bG\bC\bG\bC$.
The enzymes cut double-strands
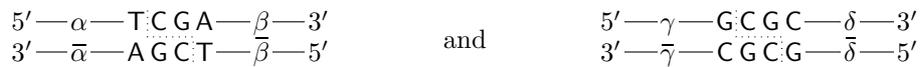
\begin{align*}
	\tikz[baseline=-.15cm,text height=1.75ex,text depth=.25ex,inner sep=1pt]{\strandA\strandB} &&
	\text{and} &&
	\tikz[baseline=-.15cm,text height=1.75ex,text depth=.25ex,inner sep=1pt]{\strandD\strandC}
\end{align*}
along the dotted lines, respectively, leaving the first segment of the left strand with a sticky end $\bG\bC$ which is compatible to the sticky end $\bC\bG$ of the second segment of the right strand.
The segments can be recombined to form either the original strands or the new strand
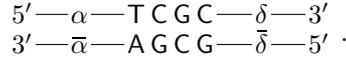
\begin{equation*}
	\tikz[baseline=-.15cm,text height=1.75ex,text depth=.25ex,inner sep=1pt]{\strandA\strandD}\ .
\]

A {\em splicing system} is a formal language model which consists of a set of {\em initial words} or {\em axioms} $I$ and a set of {\em splicing rules} $R$.
The most commonly used definition for a splicing rule
is a quadruple of words $r = (u_1,v_1;u_2,v_2)$.
This rule splices two words 
$x_1u_1v_1y_1$ and $x_2u_2v_2y_2$:
the words are cut between the factors $u_1,v_1$, respectively $u_2,v_2$, and
the prefix (the left segment) of the first word is 
recombined by catenation with the suffix
(the right segment) of the second word, see Figure~\ref{fig:splice:intro}
and also  \cite{Paun96}.
A splicing system generates a language which contains every word that can be obtained by successively applying rules to axioms and the intermediately produced words.

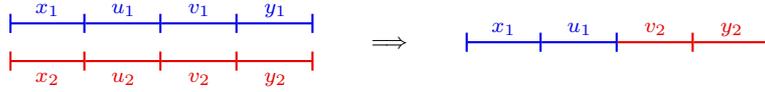
\begin{figure}[ht]
	\centering
	\begin{tikzpicture}[text height=1ex,text depth=0ex,
			font={\footnotesize}]
		\begin{scope}[yshift=.25cm,thick]
			\draw [|-|,eins] (0,0) -- node [above] {$x_1$} (1,0);
			\draw [-|,eins] (1,0) -- node [above] {$u_1$} (2,0);
			\draw [-|,eins] (2,0) -- node [above] {$v_1$} (3,0);
			\draw [-|,eins] (3,0) -- node [above] {$y_1$} (4,0);

			\draw [|-|,zwei] (0,-.5) -- node [below] {$x_2$} (1,-.5);
			\draw [-|,zwei] (1,-.5) -- node [below] {$u_2$} (2,-.5);
			\draw [-|,zwei] (2,-.5) -- node [below] {$v_2$} (3,-.5);
			\draw [-|,zwei] (3,-.5) -- node [below] {$y_2$} (4,-.5);
		\end{scope}

		\node at (5,0) {$\Longrightarrow$};

		\begin{scope}[xshift=6cm,thick]
			\draw [|-|,eins] (0,0) -- node [above] {$x_1$} (1,0);
			\draw [-|,eins] (1,0) -- node [above] {$u_1$} (2,0);
			\draw [-|,zwei] (2,0) -- node [above] {$v_2$} (3,0);
			\draw [-|,zwei] (3,0) -- node [above] {$y_2$} (4,0);
			
%
%
		\end{scope}
	\end{tikzpicture}
	\caption{Splicing of the words $x_1u_1v_1y_1$ and $x_2u_2v_2y_2$
		by the rule $r = (u_1,v_1;u_2,v_2)$.}
	\label{fig:splice:intro}
\end{figure}

\begin{example}\label{ex:first}
Consider the splicing system $(I,R)$ with axiom $I=\oneset{ab}$ and rules 	$R=\oneset{r,s}$ where $r=(a,b;\e,ab)$ and $s=(ab,\e;a,b)$; in this paper, $\e$ denotes the empty word.
Applying the rule $r$ to two copies of the axiom $ab$ creates the word $aab$ and applying the rule $s$ to two copies of the axiom $ab$ creates the word $abb$.
More generally, the rule $r$ or $s$ can be applied to words $a^ib^j$ and $a^kb^\ell$ with $i,j,k,\ell\ge 1$ in order to create the word $a^{i+1}b^\ell$ or $a^ib^{\ell+1}$, respectively.
The language generated by the splicing system $(I,R)$ is $L(I,R) = a^+b^+$.
\end{example}

The most natural variant of splicing systems, often referred to as {\em finite splicing systems}, is to consider a finite set of axioms and a finite set of rules.
In this paper, by a splicing system we always mean a finite splicing system.
Shortly after the introduction of splicing in formal language theory,
Culik~II and Harju \cite{CulikH91} proved that
splicing systems generate regular languages, only;
see also \cite{HeadP06,Pixton96}.
Gatterdam \cite{Gatterdam89} gave $(aa)^*$ as an example of a
regular language which cannot be generated by a splicing system;
thus, the class of languages generated by splicing systems is
strictly included in the class of regular languages.
However, for any regular language $L$ over an alphabet $\Sig$,
adding a marker $b\notin\Sig$ to the left side of every word in $L$ results in
the language $bL$ which can be generated by a splicing system
\cite{Head98};
\eg the language $b(aa)^*$ is generated by the axioms $\oneset{b,baa}$
and the rule $(baa,\e;b,\e)$.

This led to the question of whether or not
one of the known subclasses of the regular languages
corresponds to the class \cS of languages which can be generated by
a splicing system.
All investigations to date indicate that the class \cS
does not coincide with another naturally defined language class.
A characterization of {\em reflexive} splicing systems
using {\em Sch\"utzenberger constants} has been given by Bonizzoni, de Felice, and
Zizza \cite{BonizzoniFZ05,Bonizzoni10,BonizzoniFZ10}.
A splicing system is reflexive if for all rules $(u_1,v_1;u_2,v_2)$ in the system
we have that $(u_1,v_1;u_1,v_1)$ and $(u_2,v_2;u_2,v_2)$ are rules in the system, too.
A word $v$ is a Sch\"utzenberger constant of a language $L$ if
$x_1vy_1\in L$ and $x_2vy_2\in L$ imply $x_1vy_2\in L$ \cite{Schutzenberger75}.
Recently, it was proven by Bonizzoni and Jonoska
that every  splicing language has a constant
\cite{BonizzoniJ11}.
However, not all languages which have a constant are generated by splicing systems,
\eg in the language $L =(aa)^* + b^*$ every word $b^i$ is a constant,
but $L$ is not generated by a splicing system.

Another approach was to find an algorithm which decides whether a given regular
language is generated by a splicing system.
This problem has been investigated by Goode, Head, and Pixton
\cite{HeadPG02,GoodePHD,GoodeP07}
but it has only been partially solved:
it is decidable whether a regular language is generated by a reflexive
splicing system.
It is worth mentioning that a splicing system by the original definition in
\cite{Head87} is always reflexive.
A related problem has been investigated by Kim \cite{Kim97}: given a regular language $L$ and a finite set of {\em enzymes}, represented by set of reflexive rules $R$, Kim showed that it is decidable whether or not $L$ can be generated from a finite set of axioms by using only rules from $R$.

In this paper we settle the decidability problem, by proving that for a given regular language, it is indeed decidable whether the language is generated by a splicing system (which is not necessarily reflexive), Corollary~\ref{cor:decidability}.
More precisely, for every regular language $L$ there exists a splicing system $(I_L,R_L)$ and if $L$ is a splicing language, then $L$ is generated by the splicing system $(I_L,R_L)$.
The size of this splicing system depends on the size of the syntactic monoid of $L$.
If $m$ is the size of the syntactic monoid of $L$, then all axioms in $I_L$ and
the four components of every rule in $R_L$ have length in $\Oh(m^2)$, Theorem~\ref{thm:main:four}.
By results from \cite{HeadP06,HeadPG02}, we can construct a finite automaton which accepts the language generated by $(I_L,R_L)$, compare it with a finite automaton which accepts $L$, and thus, decide whether $L$ is generated by a splicing system.
Furthermore, we prove a similar result for a more general variant of splicing that has been introduced by Pixton \cite{Pixton96}, Theorem~\ref{thm:main}.

The paper is organized as follows.
In Section~\ref{sec:notation} we lay down the notation, recall some well-known results about syntactic monoids, and prove a pumping argument that is of importance for the proofs in the succeeding sections.
Section~\ref{sec:pixton} (Section~\ref{sec:main}) contains
the proof that a regular language $L$ is generated by
a {\em Pixton splicing system} (\resp {\em\ classical splicing system})
if and only if it is generated by one particular Pixton splicing system 
(\resp classical splicing system)
whose sice is bounded by the size of the syntactic monoid of $L$.
Sections~\ref{sec:pixton} and~\ref{sec:main}
can be read independently and overlap in some of their
main ideas.
The inclusion of both sections and the presentation order are
chiefly for expository purposes:
Due to the features of the Pixton splicing, Section~\ref{sec:pixton}
introduces the main ideas in a significantly more readable way.
Finally, in Section~\ref{sec:decidability} we deduce the decidability results
for both splicing variants.

An extended abstract of this paper, including a shortened proof of Theorem~\ref{thm:main:four} and Corollary~\ref{cor:decidability}~i.), has been published in the conference proceedings of DNA~18 in 2012 \cite{KariKopecki12dna}.
Theorem~\ref{thm:main} and Corollary~\ref{cor:decidability}~ii.)\ have not been published elsewhere.

\section{Notation and Preliminaries}
\label{sec:notation}

We assume the reader to be familiar with the fundamental
concepts of language theory, see \cite{HopcroftUllman}.

Let $\Sigma$ be a finite set of {\em letters}, the {\em alphabet};
$\Sigma^*$ be the set of all words over $\Sigma$;
and $\e$ denote the {\em empty word}.
A subset $L$ of $\Sigma^*$ is a {\em language} over $\Sigma$.
Throughout this paper, we consider languages over the fixed
alphabet $\Sig$, only.
Let $w\in \Sig^*$ be a word.
The length of $w$ is denoted by $\abs w$.
(We use the same notation for the cardinality $\abs S$ of a set $S$, as usual.)
We consider the letters of $\Sig$ to be
ordered and for words $u,v\in\Sig^*$
we denote the {\em length-lexicographical order}
by $u\eLL v$;
\ie $u\eLL v$ if either $\abs u \le \abs v$, or
$\abs u = \abs v$ and $u$ is at most $v$ in lexicographic order.
The {\em strict length-lexicographic order} is denoted by $\LL$;
we have $u\LL v$ if $u\eLL v$ and $u \neq v$.

For a length bound $m\in\N$ we let $\Sig^{\leq m}$ denote the set of words
whose length is at most $m$, \ie $\Sig^{\leq m} = \bigcup_{i\leq m}\Sig^i$.
Analogously, we define $\Sig^{<m}=\bigcup_{i<m}\Sig^i$.

If $w = xyz$ for some $x,y,z\in \Sigma^*$,
then $x$, $y$, and $z$ are called {\em prefix}, {\em factor},
and {\em suffix} of $w$,
respectively.
If a prefix or suffix of $w$ is distinct from $w$,
it is said to be {\em proper}.

Let $w = a_1\ldots a_n$ where $a_1,\ldots,a_n$ are letters from $\Sigma$.
By $w\pos{i}$ for $0\le i\le n$ we denote a {\em position} in the word $w$: 
if $i = 0$, it is the position before the first letter $a_1$, if $i = n$ it is the position after the last letter $a_n$, and otherwise, it is the position between the letters $a_i$ and $a_{i+1}$.
We want to stress that $w\pos{i}$ is not a letter in the word $w$.
By $w\pos{i;j}$ for $0\le i\le j\le n$ we denote the factor $a_{i+1}\cdots a_j$ which is enclosed by the positions $w\pos{i}$ and $w\pos{j}$.
If $x = w\pos{i;j}$ we say the factor $x$ starts at position $w\pos{i}$ and ends at position $w\pos{j}$.
Whenever we talk about a factor $x$ of a word $w$ we mean a factor starting (and ending) at a certain position, even if the the word $x$ occurs as a factor at several positions in $w$.
Let $x = w\pos{i;j}$ and $y = w\pos{i';j'}$ be factors of $w$.
We say the factors $x$ and $y$ {\em match} (in $w$) if $i = i'$ and $j=j'$;
the factor $x$ is {\em covered} by the factor $y$ (in $w$) if $i'\le i \le j \le j'$;
and the factors $x$ and $y$ {\em overlap} (in $w$) if $x\ne \e$, $y\ne\e$, and $i \le i' < j$ or $i'\le i < j'$.
In other words, if two factors $x$ and $y$ overlap in $w$, then they share a common letter of $w$.
Let $x = w\pos{i;j}$ be a factor of $w$ and let $p = w\pos{k}$ be a position in $w$.
We say the position $p$ {\em lies at the left of} $x$ if $k\le i$;
the position $p$ {\em lies at the right of} $x$ if $k\ge j$;
and the position $p$ {\em lies in} $x$ if $i < k < j$.

Every language $L$ induces an {\em syntactic congruence} $\sim_L$ over words
such that $u\sim_L v$ if and only if for all words 
$x,y$
\[
	xuy \in L \iff xvy\in L.
\]
The {\em syntactic class} (with respect to $L$) of a word $u$ is
$[u]_L = \set{v}{u\sim_L v}$.
The {\em syntactic monoid} of $L$ is the quotient monoid
\[
	M_L = \Sig^* / {\sim_L} = \set{[u]_L}{u\in \Sig^*}.
\]
It is well known that a language $L$ is regular if and only
if its syntactic monoid $M_L$ is finite.
We will use two basic facts about syntactic monoids of regular languages.

\begin{lemma}\label{lem:pumping}
	Let $L$ be a regular language
	and let $w$ be a word with $\abs w \ge \abs{M_L}^2$.
	We can factorize $w = \abg$ with $\bet\ne\e$ such that
	$\alp \sim_L \alp\bet$ and $\gam \sim_L \bet\gam$.
\end{lemma}

\begin{proof}
Consider a word $w$ with $n = \abs{w} \ge \abs{M_L}^2$.
For $i = 0,\ldots, n$, let $X_i = w\pos{0;i}$ be the syntactic classes of the prefixes of $w$ and let $Y_i = w\pos{i;n}$ be the syntactic classes of the suffixes of $w$.
Note that $X_iY_i = [w]_L$.
By the pigeonhole principle, there are $i,j$ with $0\le i <j\le n$ such that $X_i = X_j$ and $Y_i = Y_j$.
Let $\alp = w\pos{0;i}$, $\bet = w\pos{i;j}$, and $\gam = w\pos{j;n}$.
As $\alp\in X_i$ and $\alp\bet\in X_j$, we see that $\alp\sim_L \alp\bet$ and, symmetrically, $\gam\sim_L\bet\gam$.
\end{proof}

\begin{lemma}\label{lem:smallest:words}
	Let $L$ be a regular language.
	Every element $X\in M_L$ contains a word $x\in X$
	with $\abs x < \abs{M_L}$.
\end{lemma}

\begin{proof}
	We define a series of sets $S_i \sse M_L$.
	We start with $S_0 = \oneset{1}$ (here, $1 = [\e]_L$)
	and let $S_{i+1} = S_i \cup\set{X\cdot[a]_L}{X\in S_i\land a\in\Sig}$
	for $i \ge 0$.
	It is not difficult to see that
	$X\in S_i$ if and only if $X$ contains a word $x\in X$ with $\abs x \le i$.
	As $S_i\sse S_{i+1}$ and $M_L$ is finite, the series has a fixed point
	$S_n$ such that $S_i = S_n$ for all $i\ge n$.
	Let $n$ be the least value with this property, \ie $S_{n-1}\ssne S_n$ or $n=0$.
	Observe that $n < \abs{M_L}$ as $S_0 \ssne S_1\ssne \cdots \ssne S_n$.
	Every element $X\in M_L$ contains some word $w\in X$, thus,
	$X \in S_{\abs w} \sse S_n$.
	Concluding that $X$ contains
	a word with a length of at most $n< \abs{M_L}$.
\end{proof}

\subsection{A Pumping Algorithm}\label{sec:pump:alg}

Consider a regular language $L$, a word $\alp\bet\gam$ where $\alp\sim_L\alp\bet$ and $\gam\sim_L\bet\gam$, due to Lemma~\ref{lem:pumping}, and a large even number $j$.
In the proofs of Theorem~\ref{thm:main} and Lemma~\ref{lem:long:rules}, we need a pumping argument to replace all factors $\abg$ by $\abjg$ in a word $z$ in order to obtain a word $\ti z$; thus, $z\sim_L \ti z$.
As $\abg$ may be a factor of $\abjg$, we cannot ensure that $\abg$ is not a factor of $\ti z$.
However, we can ensure that if $\abg = \ti z\pos{k;k'}$ is a factor of $\ti z$, then either
\begin{inparaenum}[(a)]
	\item $\alp\bet^{j\slash 2}$ is a factor of $\ti z$ starting at position
		$\ti z \pos k$ or
	\item $\bet^{j\slash 2}\gam$ is a factor of $\ti z$ ending at position
		$\ti z \pos{k'}$;
\end{inparaenum}
\ie either $\alp$ is succeeded by a large number of $\bet$'s or $\gam$ is preceded by a large number of $\bet$'s.
The next lemma is a technical result whose purpose is to assure that for any word $z$ there exists a word $\ti z$ such that the above-mentioned property holds and $\ti z$ is generated by applying several successive pumping steps $\abg \mapsto \abjg$ to $z$.

\begin{lemma}\label{lem:pumping:algorithm}
 	Let $z,\alp,\bet,\gam$ be words with $\bet \neq \e$, let $\ell = \abs{\abg}$,
 	and let $j > \abs{z}+\ell$ be an even number.
	The following algorithm will terminate and output $\ti z$.
	\pumping{\ell}
\end{lemma}

Before we prove Lemma~\ref{lem:pumping:algorithm}, let us recall a basic fact 
about primitive words.
A word $p$ is called {\em primitive} if there is no word $x$
and $i\ge 2$ such that $p = x^i$.
The {\em primitive root} of a word $w\neq \e$ is the unique primitive
word $p$ such that $w = p^i$ for some $i\ge 1$.
For primitive $p$,
it is well known that if $pp = xpy$, then either $x=p$ and $y = \e$,
or $x=\e$ and $y = p$.
In other words, whenever $p$ is a factor of $p^n$ starting at position $p^n\pos i$,
then $i \in \abs p\cdot \N$.

For a word $w = xy$ we employ the notations $x^{-1}w = y$ and $wy^{-1} = x$.
If $x$ is not a prefix of $w$ ($y$ is not a suffix of $w$), then the $x^{-1}w$ (\resp $wy^{-1}$) is undefined.

\begin{proof}[Proof of Lemma~\ref{lem:pumping:algorithm}]
Let $p$ be the primitive root of $\bet$ and let $m$ such that $\bet = p^m$.

First, observe that if, during the computation, a factor $\abg = \ti z\pos{k;k+\ell}$ is covered by a factor $\abjg$ in $\ti z$, then either \A or \B holds.
Indeed, if $\abg = (\abjg)\pos{i;i+\ell}$ for some $i$, then $\bet$ is a factor of $\bet^j$ starting at position $\bet^j\pos{i}$.
As mentioned above, $i \in \abs{p}\cdot\N$ and either position $\bet^j[i]$ is preceded or succeeded by $p^{m\cdot j\slash2} = \bet^{j\slash2}$.
Therefore, \A or \B is satisfied.

Let $z_0 = z$, let $z_n$ be the word $\ti z$ after the $n$-th pumping step in the algorithm, and let $y = p^{m\cdot j-2} = \bet^jp^{-2}$.
For each $n$, we will define a unique factorization
\[
	z_n = x_{n,0} y x_{n,1} \cdots  y x_{n,n}
\]
where $p$ is a suffix of $x_{n,i}$ for $i = 0,\ldots, n-1$ and $p$ is a prefix of $x_{n,i}$ for $i = 1,\ldots,n$.
This factorization is defined inductively:
naturally, we start with $x_{0,0} = z_0 = z$.
Assume $z_n$ is factorized in the above manner.
Let $\abg = z_n\pos{k;k+\ell}$ be the factor, such that neither \A nor \B holds, which we replace in the $(n+1)$-st step (if there is no such factor, the algorithm terminates and we do not have to define $z_{n+1}$).
By contradiction, assume that $\alp$ starting at position $z_n\pos{k}$ is covered by the $i$-th factor $y = p^{m\cdot j-2}$ in the factorization of $z_n$ for some $1\le i\le n$.
By the first observation, the factor $\bet\gam = z_n\pos{k+\abs\alp;k+\ell}$ must overlap with $x_i$.
However, as $p$ is a prefix of $x_i$, the factor $\bet = z_n\pos{k+\abs\alp;k+\abs{\alp\bet}}$ has to cover the prefix $p$ of $x_i$ or it has to cover one of the $p$'s in $y$.
This implies that $\gam$ is preceded by $p^{m\cdot j\slash2} = \bet^{j\slash2}$ and \B holds --- contradiction.
Symmetrically, $\gam$ is not covered by one of the factors $y$ in $z_n$ neither.
	
Thus, $\bet = z_n\pos{k+\abs\alp;k+\abs{\alp\bet}}$ is covered by some $x_{n,i}$ in the factorization of $z_n$ and $x_{n,i}$ can be factorized $x_{n,i} = u\bet v$ where $u\neq \e$ and $v \neq \e$.
Note that the length of $x_{n,i}$ has to be at least $\abs{\bet}+2$.
Now, let $x_{n+1,h} = x_{n,h}$ for $h = 0,\ldots,i-1$, let $x_{n+1,h+1} = x_{n,h}$ for $h = i+2,\ldots,n$, let $x_{n+1,i} = up$, and let $x_{n+1,i+1} = pv$.
Observe that this defines the desired factorization.
Also note that 
\[
	\abs{x_{n+1,i}} = \abs u + \abs p
		= \abs{x_{n,i}} -\abs{\bet} -\abs{v} +\abs{p}
		\le \abs{x_{n,i}} -\abs{v} < \abs{x_{n,i}}
\]
and, symmetrically, $\abs{x_{n+1,i+1}} < \abs{x_{n,i}}$.
Thus, in each pumping step, we replace one of the factors $x_{n,i}$ by two strictly shorter factors $x_{n+1,i}$ and $x_{n+1,i+1}$.
As we have noted above, in a factor $x_{n,i}$ cannot be pumped anymore, if it is shorter than $\abs\bet +2$.
Eventually, all the the factors will be too short and the pumping algorithm will stop.
\end{proof}

\section{Pixton's Variant of Splicing}
\label{sec:pixton}

In this section we use the definition of the splicing operation as it 
was introduced in \cite{Pixton96}.
A triplet of words $r = (u_1,u_2;v)\in (\Sig^*)^3$ is called a 
{\em (splicing) rule}.
The words $u_1$ and $u_2$ are called {\em left} and {\em right site} of $r$,
respectively,
and $v$ is the {\em bridge} of $r$.
This splicing rule can be applied to two words 
$w_1 = x_1u_1y_1$ and $w_2 = x_2u_2y_2$, that each contain one of the sites,
in order to create the new word $z = x_1vy_2$,
see Figure~\ref{fig:splicing}.
This operation is called {\em splicing} and
it is denoted by $(w_1,w_2)\splice{r} z$.

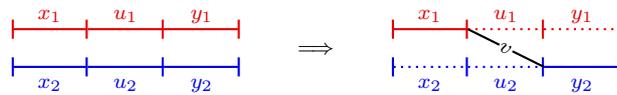
\begin{figure}[ht]
	\centering
	\begin{tikzpicture}[text height=1ex,text depth=0ex,
			font={\footnotesize},thick]
		\begin{scope}
			\draw [|-|,zwei] (0,0) -- node [above] {$x_1$} (1,0);
			\draw [-|,zwei] (1,0) -- node [above] {$u_1$} (2,0);
			\draw [-|,zwei] (2,0) -- node [above] {$y_1$} (3,0);

			\draw [|-|,eins] (0,-.5) -- node [below] {$x_2$} (1,-.5);
			\draw [-|,eins] (1,-.5) -- node [below] {$u_2$} (2,-.5);
			\draw [-|,eins] (2,-.5) -- node [below] {$y_2$} (3,-.5);
		\end{scope}

		\node at (4,-.25) {$\Longrightarrow$};

		\begin{scope}[xshift=5cm]
			\draw [|-|,zwei] (0,0) -- node [above] {$x_1$} (1,0);
			\draw [-|,dotted,zwei] (1,0) -- node [above] {$u_1$} (2,0);
			\draw [-|,dotted,zwei] (2,0) -- node [above] {$y_1$} (3,0);

			\draw [|-|,dotted,eins] (0,-.5) -- node [below] {$x_2$} (1,-.5);
			\draw [-|,dotted,eins] (1,-.5) -- node [below] {$u_2$} (2,-.5);
			\draw [-|,eins] (2,-.5) -- node [below] {$y_2$} (3,-.5);

			\draw (1,0) -- node [sloped,fill=white,inner sep=.5pt]
				{$v$} (2,-.5);
		\end{scope}
	\end{tikzpicture}
	\caption{Splicing of the words $x_1u_1y_1$ and $x_2u_2y_2$ by the rule
		$r = (u_1,u_2;v)$.}
	\label{fig:splicing}
\end{figure}

For a rule $r$ we define the {\em splicing operator} $\sig_r$
such that for a language $L$
\[
	\sig_r(L) = 
		\set{z\in\Sig^*}{\exists w_1,w_2\in L \colon (w_1,w_2)\splice{r} z}
\]
and for a set of splicing rules $R$, we let
\[
	\sig_R(L) = \bigcup_{r\in R} \sig_r(L).
\]
The reflexive and transitive closure of the splicing operator $\sig_R^*$
is given by
\begin{align*}
	\sig_R^0(L) &= L, &
	\sig_R^{i+1}(L) &= \sig_R^i(L) \cup \sig_R(\sig_R^i(L)), &
	\sig_R^*(L) &= \bigcup_{i\ge 0} \sig_R^i(L).
\end{align*}
A finite set of axioms $I \sse \Sig^*$ and a
finite set of splicing rules $R \sse (\Sig^*)^3$
form a {\em splicing system} $(I,R)$.
Every splicing system $(I,R)$ generates a language
$L(I,R) = \sig_R^*(I)$.
Note that $L(I,R)$ is the smallest language which is closed under
the splicing operator $\sig_R$ and includes $I$.
It is known that the language generated by a splicing system
is regular, see \cite{Pixton96}.
A (regular) language $L$ is called a {\em splicing language}
if a splicing system $(I,R)$ exists such that $L = L(I,R)$.

A rule $r$ is said to {\em respect} a language $L$
if $\sig_r(L) \sse L$.
It is easy to see that for any splicing system $(I,R)$,
every rule $r \in R$ respects the generated language $L(I,R)$.
Moreover, a rule $r\notin R$ respects $L(I,R)$ if and only if
\mbox{$L(I,R\cup\oneset r)={}$}\mbox{$L(I,R)$}.
We say a splicing $(w_1,w_2)\splice r z$ {\em respects} 
a language $L$ if  $w_1,w_2\in L $ and $r$ respects $L$;
obviously, this implies $z\in L$, too.

Pixton introduced this variant of splicing in order to give a simple proof
for the regularity of languages generated by splicing systems.
As Pixton's variant of splicing is more general than the
{\em classic splicing}, defined in the introduction and in Section~\ref{sec:main},
his proof of regularity also applies to classic splicing systems.
For a moment, let us call a classic splicing rule a
{\em quadruple} and a Pixton splicing rule a {\em triplet}.
Consider a quadruplet $r = (u_1,v_1;u_2,v_2)$.
It is easy to observe that whenever we can use $r$ in order to splice
$w_1 = x_1u_1v_1y_1$ with $w_2 = x_2u_2v_2y_2$ to obtain the word
$z = x_1u_1v_2y_2$,
we can use the triplet $s = (u_1v_1,u_2v_2;u_1v_2)$ in order to splice
$(w_1,w_2)\splice{s} z$  as well.
However, for a triplet $s = (u_1,u_2;v)$ where $v$ is not a concatenation
of a prefix of $u_1$ and a suffix of $u_2$,
there is no quadruplet $r$ that can be used for the same splicings.
Moreover, the class of classical splicing languages is strictly included in the class of Pixton splicing languages; \eg the language 
\[
	L=cx^*ae+cx^*be+dcx^*bef
\]
over the alphabet $\oneset{a,b,c,d,e,f,x}$ is a Pixton splicing language but not a classical splicing language, see \cite{BonizzoniFMZ01}.
For the rest of this section we focus on Pixton's splicing variant
and by a rule we always mean a triplet.

The main result of this section states that if a regular language $L$ is a splicing language, then it is created by a particular splicing system $(I,R)$ which only depends on the syntactic monoid of $L$.

\begin{theorem}\label{thm:main}
	Let $L$ be a splicing language and $m = \abs{M_L}$.
	The splicing system $(I,R)$ with $I = \Sig^{<m^2+ 6m}\cap L$ and
	\[
		R = \set{r\in\Sig^{< 2m}\times\Sig^{< 2m}\times\Sig^{< m^2+10m}}
		{r\text{ respects }L}
	\]
	generates the language $L = L(I,R)$.
\end{theorem}

As the language generated by the splicing system $(I,R)$ is constructible,
Theorem~\ref{thm:main} implies that the problem whether or not a given
regular language is a splicing language is decidable.
A detailed discussion of the decidability result 
is given in Section~\ref{sec:decidability}.

Let $L$ be a formal language.
Clearly, every set of words $J\sse L$ and set of rules $S$ where every rule in $S$ respects $L$ generates a subset $L(J,S)\sse L$.
Therefore, in Theorem~\ref{thm:main} the inclusion $L(I,R) \sse L$ is obvious.
The rest of this section is devoted to the proof of the converse inclusion $L\sse L(I,R)$.
Consider a splicing language $L$.
One of the main techniques we use in the proof is that,
whenever a word $z$ is created by a series of splicings from
a set of words in $L$ and a set rules that respect $L$,
then we can use a modified set of words from $L$ and modified rules which
respect $L$ in order to obtain the same word $z$ by splicing.
If $z$ is sufficiently long these words can be chosen such that they are all shorter
than $z$ and the sites and bridges of the rules also satisfy certain length
restrictions.
Of course, our goal is to show that we can create $z$ by splicing from a subset of $I$
with rules which all satisfy the length bounds given by $R$
(as defined in Theorem~\ref{thm:main}).
In Section~\ref{sec:rules} we will present techniques to obtain rules that respect
$L$ from other rules respecting $L$ and we show how we can modify a single
splicing step, such that the words used for splicing are not significantly longer
than the splicing result.
In Section~\ref{sec:series} we use these techniques to modify series of splicings
in the way described above (Lemma~\ref{lem:series}).
Finally, in Section~\ref{sec:proof} we prove Theorem~\ref{thm:main}.

\subsection{Rule Modifications}\label{sec:rules}

Let us start with the simple observation that we can extend the sites and the
bridge of a rule $r$ such that the new rule respects all languages
which are respected by $r$.

\begin{lemma}\label{lem:rule:ext}
	Let $r=(u_1,u_2;v)$ be a rule which respects a language $L$.
	For every word $x$,
	the rules $(xu_1,u_2;xv)$, $(u_1x,u_2;v)$,
	$(u_1,xu_2;v)$, and $(u_1,u_2x;vx)$ respect $L$ as well.
\end{lemma}

\begin{proof}
	Let $s$ be any of the four rules
	$(xu_1,u_2;xv)$, $(u_1x,u_2;v)$,
	$(u_1,xu_2;v)$, or $(u_1,u_2x;vx)$.
	In order to prove that $s$ respects $L$
	we have to show that, 
	for all $w_1,w_2\in L$ and $z\in\Sig^*$
	such that $(w_1,w_2) \splice s z$, we have $z\in L$, too.
	Indeed, if $(w_1,w_2) \splice s z$, then
	$(w_1,w_2) \splice r z$ and as $r$ respects $L$,
	we conclude $z\in L$.
\end{proof}

Henceforth, we will refer to the rules $(xu_1,u_2;xv)$ and $(u_1,u_2x;vx)$
as extensions of the bridge and to the rules $(u_1x,u_2;v)$ and $(u_1,xu_2;v)$
as extensions of the left and right site, respectively.

Next, for a language $L$,
let us investigate the syntactic class of a rule $r = (u_1,u_2;v)$.
The {\em syntactic class} (with respect to $L$) of $r$ is the set
of rules $[r]_L = [u_1]_L \times [u_2]_L\times [v]_L$ and 
two rules $r$ and $s$ are {\em syntactically congruent} (with respect to $L$),
denoted by $r \sim_L s$, if $s \in [r]_L$.

\begin{lemma}\label{lem:rule:sim}
	Let $r$ be a rule which respects a language $L$.
	Every rule $s\in[r]_L$ respects $L$.
\end{lemma}

\begin{proof}
	Let $r = (u_1,u_2; v)$ and $s = (\ti u_1,\ti u_2; \ti v)$.
	Thus, $u_i\sim_L \ti u_i$ for $i= 1,2$ and $v\sim_L \ti v$.
	For $\ti w_1 = x_1\ti u_1y_1\in L$ and $\ti w_2 = x_2\ti u_2y_2\in L$
	we have to show that $\ti z = x_1\ti v y_2\in L$.
	For $i = 1,2$, let $w_i = x_iu_iy_i$ and note that $w_i \sim_L \ti w_i$;
	hence, $w_i\in L$.
	Furthermore, $(w_1,w_2)\splice r x_1 vy_2 = z \in L$ as $r$ respects $L$
	and $\ti z\in L$ as $z\sim_L \ti z$.
\end{proof}

Consider a splicing $(x_1u_1y_1,x_2u_2y_2)\splice r x_1vy_2$ which respects
a regular language $L$
as shown in Figure~\ref{fig:shorten:sites} left side.
The factors $u_1y_1$ and $x_2u_2$ may be relatively long but they do
not occur as factors in the resulting word $x_1vy_2$.
In particular, it is possible that two long words are spliced 
and the outcome is a relatively short word.
Using Lemmas~\ref{lem:rule:ext} and~\ref{lem:rule:sim},
we can find shorter words in $L$ and a modified splicing rule
which can be used to obtain $x_1vy_2$.

\begin{figure}[ht]
	\center
	\begin{tikzpicture}[text height=1ex,text depth=0ex,
			font={\footnotesize},thick]
		\begin{scope}
			\draw [|-|,zwei] (0,0) -- node [above] {$x_1$} (1,0);
			\draw [-|,dotted,zwei] (1,0) -- node [above] {$u_1$} (2,0);
			\draw [-|,dotted,zwei] (2,0) -- node [above] {$y_1$} (3,0);

			\draw [|-|,dotted,eins] (0,-.5) -- node [below] {$x_2$} (1,-.5);
			\draw [-|,dotted,eins] (1,-.5) -- node [below] {$u_2$} (2,-.5);
			\draw [-|,eins] (2,-.5) -- node [below] {$y_2$} (3,-.5);

			\draw (1,0) -- node [sloped,fill=white,inner sep=.5pt]
				{$v$} (2,-.5);
		\end{scope}

		\node at (4,-.25) {$\Longrightarrow$};

		\begin{scope}[xshift=5cm]
			\draw [|-|,zwei] (0,0) -- node [above] {$x_1$} (1,0);
			\draw [-|,dotted,zwei] (1,0) -- node [above] {$\ti u_1$} (1.75,0);

			\draw [|-|,dotted,eins] (1.25,-.5) --
				node [below] {$\ti u_2$} (2,-.5);
			\draw [-|,eins] (2,-.5) -- node [below] {$y_2$} (3,-.5);

			\draw (1,0) -- node [sloped,fill=white,inner sep=.5pt]
				{$v$} (2,-.5);
		\end{scope}
	\end{tikzpicture}
	\caption{The factors $u_1y_1$ and $x_2u_2$ can be replaced
		by {\em short} words.}
	\label{fig:shorten:sites}
\end{figure}
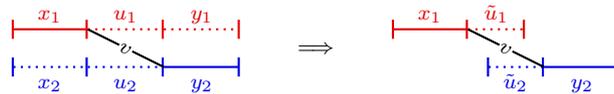

\begin{lemma}\label{lem:shorten:sites}
	Let $r = (u_1,u_2;v)$ be a rule which respects a regular language $L$
	and $w_1 = x_1u_1y_1\in L$, $w_2 = x_2u_2y_2\in L$.
	There is a rule $s = (\ti u_1,\ti u_2;v)$ which respects $L$
	and words $\ti w_1 = x_1\ti u_1 \in L$,
	$\ti w_2 = \ti u_2 y_2\in L$
	such that $\abs{\ti u_1}, \abs{\ti u_2} < \abs{M_L}$.
	More precisely, $\ti u_1\in [u_1y_1]_L$ and $\ti u_2\in [x_2u_2]_L$.
	
	In particular, whenever $(w_1,w_2)\splice r x_1vy_2 = z$,
	there is a splicing $(\ti w_1,\ti w_2)\splice s z$
	which respects $L$ where $\ti w_1$, $\ti w_2$, and $s$ have the properties
	described above.
\end{lemma}

\begin{proof}
	By Lemma~\ref{lem:rule:ext}, the rule
	$(u_1y_1,x_2u_2;v)$ respects $L$.
	Choose $\ti u_1\in [u_1y_1]_L$ and $\ti u_2\in [x_2u_2]_L$
	as shortest words from the syntactic classes, respectively;
	as such, $\abs{\ti u_1}, \abs{\ti u_2} < \abs{M_L}$ (Lemma~\ref{lem:smallest:words}) and
	$\ti w_1 = x_1\ti u_1\in L$, $w_2 = \ti u_2y_2\in L$.
	Furthermore, by Lemma~\ref{lem:rule:sim}, 
	$s = (\ti u_1,\ti u_2;v)$ respects $L$.
\end{proof}

Another way of modifying a splicing $(w_1,w_2)\splice r z$
is to extend the bridge of $r$  to the left until it covers a prefix of $w_1$.
Afterwards, we can use the same method we used in Lemma~\ref{lem:shorten:sites}
and replace $w_1$ by a short word, see Figure~\ref{fig:extend:bridge}.
As the splicing operation is symmetric,
we can also extend the bridge of $r$ rightwards and replace $w_2$
by a short word,
even though Lemma~\ref{lem:extend:bridge} does not explicitly state this.

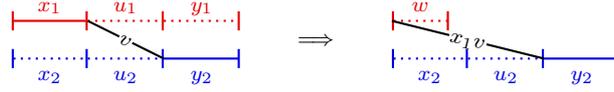
\begin{figure}[ht]
	\center
	\begin{tikzpicture}[text height=1ex,text depth=0ex,
			font={\footnotesize},thick]
		\begin{scope}
			\draw [|-|,zwei] (0,0) -- node [above] {$x_1$} (1,0);
			\draw [-|,dotted,zwei] (1,0) -- node [above] {$u_1$} (2,0);
			\draw [-|,dotted,zwei] (2,0) -- node [above] {$y_1$} (3,0);

			\draw [|-|,dotted,eins] (0,-.5) -- node [below] {$x_2$} (1,-.5);
			\draw [-|,dotted,eins] (1,-.5) -- node [below] {$u_2$} (2,-.5);
			\draw [-|,eins] (2,-.5) -- node [below] {$y_2$} (3,-.5);

			\draw (1,0) -- node [sloped,fill=white,inner sep=.5pt]
				{$v$} (2,-.5);
		\end{scope}

		\node at (4,-.25) {$\Longrightarrow$};

		\begin{scope}[xshift=5cm]
			\draw [|-|,dotted,zwei] (0,0) -- node [above] {$\ti w$} (.75,0);

			\draw [|-|,dotted,eins] (0,-.5) -- node [below] {$x_2$} (1,-.5);
			\draw [-|,dotted,eins] (1,-.5) -- node [below] {$u_2$} (2,-.5);
			\draw [-|,eins] (2,-.5) -- node [below] {$y_2$} (3,-.5);

			\draw [-] (0,0) -- node [sloped,fill=white,inner sep=.5pt]
				{$x_1v$} (2,-.5);
		\end{scope}
	\end{tikzpicture}
	\caption{The word $x_1u_1y_1$ can be replaced by a {\em short} word
	as long as we extend the bridge of the splicing rule accordingly.}
	\label{fig:extend:bridge}
\end{figure}

\begin{lemma}\label{lem:extend:bridge}
	Let $r = (u_1,u_2;v)$ be a rule which respects a regular language $L$
	and let $w_1 = x_1u_1y_1 \in L$.
	Every rule $s = (\ti w,u_2;x_1v)$,
	where $\ti w\in [w_1]_L \sse L$, respects $L$.
	In particular, there is a rule $s$, as above,
	where $\abs{\ti w} < \abs{M_L}$.
\end{lemma}

\begin{proof}
	By Lemma~\ref{lem:rule:ext}, we see that $(x_1u_1y_1,u_2;x_1v)$ respects $L$
	and, by Lemma~\ref{lem:rule:sim}, $ s=(\ti w,u_2;x_1v)$ respects $L$.
	If $\ti w\in [w_1]_L$ is a shortest word from the set, then
	$\abs{\ti w} < \abs{M_L}$ by Lemma~\ref{lem:smallest:words}.
\end{proof}

\subsection{Series of Splicings}\label{sec:series}

We are now investigating words which are created by a series of
successive splicings which all respect a regular language $L$.
Observe, that if a word is created by two (or more) successive splicings, but the bridges of the rules do not overlap in the generated word, then the order of these splicings is irrelevant.
The notation in Remark~\ref{rem:swap:Pixton} is the same as in the Figure~\ref{fig:swap:Pixton}.

\begin{figure}[ht]
	\center
	\begin{tikzpicture}[text height=1ex,text depth=0ex,
			font={\footnotesize},thick]
		\draw [|-|,zwei] (0,0) -- node [above] {$x_1$} (1,0);
		\draw [-|,dotted,zwei] (1,0) -- node [above] {$u_1$} (2,0);
		\draw [-|,dotted,zwei] (2,0) -- node [above] {$y_1$} (3,0);

		\draw [|-|,dotted,eins] (0,-.5) -- node [below] {$x_2$} (1,-.5);
		\draw [-|,dotted,eins] (1,-.5) -- node [below] {$u_2$} (2,-.5);
		\draw [-|,eins] (2,-.5) -- node [above] {$w'$} (3,-.5);
		\draw [-|,eins,dotted] (3,-.5) -- node [above] {$u_2'$} (4,-.5);
		\draw [-|,eins,dotted] (4,-.5) -- node [above] {$y_2$} (5,-.5);

		\draw [|-|,dotted,drei] (2,-1) -- node [below] {$x_3$} (3,-1);
		\draw [-|,dotted,drei] (3,-1) -- node [below] {$u_3$} (4,-1);
		\draw [-|,drei] (4,-1) -- node [below] {$y_3$} (5,-1);

		\draw (1,0) -- node [sloped,fill=white,inner sep=.5pt]
			{$v_1$} (2,-.5);
		\draw [gray] (3,-.5) -- node [sloped,fill=white,inner sep=.5pt]
			{$v_2$} (4,-1);
	\end{tikzpicture}
	\caption{The word $x_1v_1 w' v_2 y_3$ can be created either by using the
		right splicing first or by using the left splicing first.}
	\label{fig:swap:Pixton}
\end{figure}

\begin{remark}\label{rem:swap:Pixton}
Consider rules $r = (u_1,u_2;v_1)$ and $s = (u_2',u_3;v_2)$
and words $w_1 = x_1u_1y_1$, $w_2 = x_2u_2 w' u_2'y_2$,
and $w_3 = x_3u_3y_3$.
The word $z = x_1v_1 w' v_2 y_3$ can be obtained by the splicings
\begin{align*}
	(w_1,w_2) &\splice{r}  x_1v_1w'u_2'y_2 =  z', &
	(z',w_3) &\splice{s} z \qquad\text{as well as} \\
	(w_2,w_3) &\splice{s}  x_2u_2w'v_2y_3 = z'', &
	(w_1,z'') &\splice{r} z,
\end{align*}
which makes the order of the splicing steps irrelevant.
\end{remark}

Now, consider a word $z$ which is created by two successive splicings
from words $w_i = x_iu_iy_i$ for $i = 1,2,3$ as in Figure~\ref{fig:combine}.
If no factor of $w_1$ or of the bridge in the first
splicing is a part of $z$,
then we can find another splicing rule $s$ such that
$(w_3,w_2)\splice s z$
and the bridge of $s$ is the bridge used in the second splicing.

\begin{figure}[ht]
	\center
	\begin{tikzpicture}[text height=1ex,text depth=0ex,
			font={\footnotesize},thick]
		\begin{scope}
			\draw [|-|,zwei] (0,0) -- node [above] {$x_1$} (1,0);
			\draw [-|,dotted,zwei] (1,0) -- node [above] {$u_1$} (2,0);
			\draw [-|,dotted,zwei] (2,0) -- node [above] {$y_1$} (3,0);

			\draw [|-|,dotted,eins] (0,-.5) -- node [below] {$x_2$} (1,-.5);
			\draw [-|,dotted,eins] (1,-.5) -- node [below] {$u_2$} (2,-.5);
			\draw [-|,eins] (2,-.5) -- node [below] {$y_2$} (3,-.5);

			\draw (1,0) -- node [sloped,fill=white,inner sep=.5pt]
				{$v_1$} (2,-.5);
		\end{scope}

		\node at (3.5,-.25) {$+$};

		\begin{scope}[xshift=3.875cm]
			\draw [|-|,drei] (.25,0) -- node [above] {$x_3$} (1.25,0);
			\draw [-|,dotted,drei] (1.25,0) -- 
				node [above] {$u_3$} (2.25,0);
			\draw [-|,dotted,drei] (2.25,0) -- 
				node [above] {$y_3$} (3.25,0);

			\draw [|-|,dotted,zwei] (0,-.5) -- node [below] {$x_1$} (1,-.5);
			\draw [-,dotted] (1,-.5) -- node [below] {$v_1$} (2,-.5);
			\draw [|-,dotted,eins] (2,-.5) -- (2.25,-.5);
			\draw [-|,eins] (2.25,-.5) -- node [pos=.3333,below] {$y_2$} (3,-.5);

			\draw [gray] (1.25,0) -- node [sloped,fill=white,inner sep=.5pt]
				{$v_2$} (2.25,-.5);
		\end{scope}

		\node at (8,-.25) {$\Longrightarrow$};

		\begin{scope}[xshift=8.875cm]
			\draw [|-|,drei] (.25,0) -- 
				node [above] {$x_3$} (1.25,0);
			\draw [-|,dotted,drei] (1.25,0) -- 
				node [above] {$u_3$} (2.25,0);
			\draw [-|,dotted,drei] (2.25,0) -- 
				node [above] {$y_3$} (3.25,0);

			\draw [|-|,dotted,eins] (0,-.5) -- node [below] {$x_2$} (1,-.5);
			\draw [-|,dotted,eins] (1,-.5) -- node [below] {$u_2$} (2,-.5);
			\draw [-,dotted,eins] (2,-.5) -- (2.25,-.5);
			\draw [-|,eins] (2.25,-.5) -- node [pos=.3333,below] {$y_2$} (3,-.5);

			\draw [gray] (1.25,0) -- node [sloped,fill=white,inner sep=.5pt]
				{$v_2$} (2.25,-.5);
		\end{scope}
	\end{tikzpicture}
	\caption{Two successive splicings can be replaced by one splicing
		in the case when the factor $x_1$ and the bridge $v_1$ do not contribute
		to the resulting word.}
	\label{fig:combine}
\end{figure}
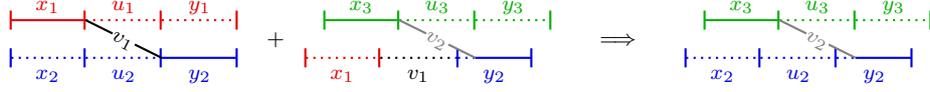

\begin{lemma}\label{lem:combine}
	Let $L$ be a language,
	$w_i = x_i u_i y_i \in L$ for $i = 1,2,3$,
	and $r_1 = (u_1,u_2;v_1)$, $r_2 = (u_3,u_4;v_2)$ be rules
	respecting $L$.
	If there are splicings
	\begin{align*}
		(w_1,w_2)&\splice{r_1} x_1v_1y_2 = w_4 = x_4u_4y_4, &
		(w_3,w_4)&\splice{r_2} x_3v_2 y_4 = z
	\end{align*}
	where $y_4$ is a suffix of $y_2$,
	then there is a rule $s = (u_3,\ti u_2;v_2)$
	which respects $L$ and $(w_3,w_2)\splice{s} z$.
\end{lemma}

\begin{proof}
By extending the bridge $v_1$ of $r_1$ and the right site $u_4$ of $r_2$ (Lemma~\ref{lem:rule:ext}), we may assume the factors $v_1$ and $u_4$ match in $w_4$:
let $w_4\pos{i,j} = v_1$ and $w_4\pos{i',j'} = u_4$,
\begin{itemize}
	\item if $i < i'$ we extend $u_4$ in $r_2$ to the left by $i'-i$ letters,
	\item if $i > i'$ we extend $v_1$ in $r_1$ to the left by $i-i'$ letters and we extend $u_1$ accordingly, and
	\item we extend $v_1$ in $r_1$ to the right by $j'-j$ letters and we extend $u_2$ accordingly;
	Note that $j'\ge j$ as $y_4$ is a suffix of $y_2$.
\end{itemize}
Clearly, the extended factors $v_1$ and $u_4$ match in $w_4$.
The left site $u_3$ and the bridge $v_2$ of $r_2$ are not modified by this extension.
Additionally, we have $x_1 = x_4$ and $y_2 = y_4$.
Let $s = (u_3,u_2;v_2)$ (where $u_2$ is the extended right site of $r_1$).
As desired, $(w_3,w_2)\splice s x_3v_2 y_4 = z$ since $w_2 = x_2u_2 y_4$.
	
Next, let us prove that $s$ respects $L$.
Let $w_i' = x_i' u_i y_i'\in L$ for $i = 2,3$.
If for all those words $x_3' v_2 y_2'\in L$, then $s$ respects $L$.
Indeed, we may splice
\begin{align*}
	(w_1,w_2')&\splice{r_1} x_1v_1 y_2' = x_1 u_4 y_2', &
	(w_3',x_1 u_4 y_2')&\splice{r_2} x_3' v_2 y_2'.
\end{align*}
Therefore, $x_3' v_2 y_2'\in L$ and $s$ respects $L$.
\end{proof}

Consider a splicing system $(J,S)$ and its generated language $L = L(J,S)$.
Let $n$ be the length of the longest word in $J$ and let $\mu$ be the length-lexicographic largest word that is a component of a rule in $S$.
Define $W_\mu = \set{w\in\Sig^*}{w\eLL \mu}$ as the set of all words that are at most as large as $\mu$, in length-lexicographical order.
Furthermore, let $I = \Sig^{\le n}\cap L$ be a set of axioms and let
\[
	R = \set{r\in W_\mu^3}{r\text{ respects }L}
\]
be a set of rules.
It is not difficult to see that $J\sse I$, $S\sse R$, and $L = L(I,R)$.
Whenever convenient, we may assume that a splicing language $L$
is generated by a splicing system which is of the form of $(I,R)$.

Now, consider the creation of a word $xzy\in L$
by splicing in $(I,R)$.
The creation of $xzy$ can be traced back to a word $z_1 = x_1zy_1$
where either $z_1\in I$ or where $z_1$ is created by a splicing
that affects $z$, \ie the bridge in this splicing
overlaps with the factor $z$ in $x_1zy_1$.
The next lemma describes this creation of $xzy = z_{k+1}$ by $k$ splicings in $(I,R)$,
and shows that we can choose the rules and words which are used to create
$z_{k+1}$ from $z_1$ such that the words and bridges of rules
are not significantly longer than $\ell = \max\oneset{\abs x,\abs y}$.

\begin{lemma}\label{lem:series}
	Let $L$ be a splicing language, let $\ell,n\in\N$, let $m = \abs{M_L}$,
	and let $\mu$ be a word with $\abs\mu \ge \ell+2m$
	such that for $I = \Sig^{\le n}\cap L$ and
	$R = \set{r\in W_\mu^3}{r\text{ respects }L}$
	we have $L = L(I,R)$.
		
	Let $z_{k+1} = x_{k+1}zy_{k+1}$, with $\abs{x_{k+1}},\abs{y_{k+1}}\le\ell$,
	be a word that is created by $k$ splicings from a word $z_1 = x_1 z y_1$
	where either $z_1\in I$ or $z_1$ is created
	by a splicing $(w_0,w_0')\splice{s} z_1$
	with $w_0,w_0'\in L$, $s\in R$,
	and the bridge of $s$ overlaps with $z$ in $z_1$.
	Furthermore, for $i = 1,\ldots,k$ the intermediate splicings are either
	\begin{enumerate}[(i)]
		\item $(w_i,z_i)\splice{r_i} x_{i+1}zy_{i+1}=z_{i+1}$,
			$w_i\in L$, $r_i\in R$, $y_{i+1} = y_i$,
			and the bridge of $r_i$
			is covered by the prefix $x_{i+1}$ or
		\item $(z_i,w_i)\splice{r_i} x_{i+1}zy_{i+1}=z_{i+1}$,
			$w_i\in L$, $r_i\in R$, $x_{i+1} = x_i$,
			and the bridge of $r_i$
			is covered by the suffix $y_{i+1}$.
	\end{enumerate}
	There are rules and words creating $z_{k+1}$, as above,
	satisfying in addition:
	\begin{enumerate}
		\item There is $k' \le k$ such that for $i = 1,\ldots,k'$
			all splicings are of the form \I and
			for $i = k'+1,\ldots,k$ all splicings are of the form \II.
		\item
			For $i = 1,\ldots,k$ the following bounds apply:
			$\abs{x_i},\abs{y_i} < \ell+2m$, $\abs{w_i} < m$,
			$r_i \in \Sig^{<2m}\times\Sig^{<2m}\times \Sig^{<\ell+m}$.
	\end{enumerate}
	In particular, if $n \ge m$, then $w_1,\ldots,w_k\in I$.
\end{lemma}

\begin{proof}
Statement~1 follows by Remark~\ref{rem:swap:Pixton}
Note that if $k = 0$, then statement~2 is trivially true.
By the first statement, $x_{k'+1} = x_{k'+2} = \cdots = x_{k+1}$
and $y_1 = y_2 =\cdots = y_{k'+1}$.
Let us consider the splicings of the form \I which are
the steps $i = 1,\ldots,k'$.
The notations we employ in order to prove the second statement
for $i = 1,\ldots,k'$
are chosen to match the notations in Figure~\ref{fig:series:proof}.

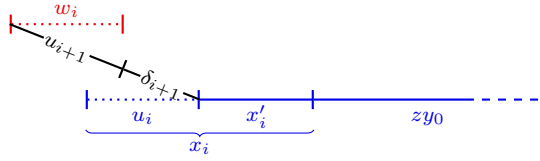
\begin{figure}[ht]
	\centering
	\begin{tikzpicture}[text height=1ex,text depth=0ex,
			font={\footnotesize}]
		\begin{scope}[thick]
			\draw [|-|,dotted,zwei] (0,0) -- node [above] {$w_i$} (1.5,0);

			\draw [|-|,dotted,eins] (1,-1) -- node [below] {$u_i$} (2.5,-1);
			\draw [-|,eins] (2.5,-1) -- node [below] {$x_i'$} (4,-1);
			\draw [-,eins] (4,-1) -- node [below,pos=.75] {$zy_0$} (6,-1);
			\draw [-,dashed,eins] (6,-1) -- (7,-1);

			\draw [-|] (0,0) -- node [sloped,fill=white,inner sep=.5pt]
				{$u_{i+1}$} (1.5,-.6);
			\draw (1.5,-.6) -- node [sloped,fill=white,inner sep=.5pt]
				{$\del_{i+1}$} (2.5,-1);
		\end{scope}
		\draw [decorate,decoration=brace,eins] (4,-1.4) to 
			node [below] {$x_i$} (1,-1.4);
	\end{tikzpicture}
	\caption{The $i$-th splicing step in the proof of Lemma~\ref{lem:series}
		where $v_i = u_{i+1}\del_{i+1}$
		and $x_{i+1} = u_{i+1}\del_{i+1} x_i'$.}
	\label{fig:series:proof}
\end{figure}

Let $r_i = (w_i,u_i;v_i)$
where $w_i \in \Sig^{<m}\cap L$ (Lemma~\ref{lem:extend:bridge})
and $x_i = u_i x_i'$;
(by Lemma~\ref{lem:rule:ext},
we extended the site $u_i$ to cover a prefix of $x_i$)
such that $u_{i+1} x_{i+1}' = v_i x_i'$
with $u_{k'+1} = \e$ and $ x_{k'+1}' = x_{k'+1} = x_{k+1}$.
Lemma~\ref{lem:combine} justifies the assumption that
every splicing occurs at the left of the preceding splicing,
\ie $x_i'$ is a proper suffix of $x_{i+1}'$.
Note that, as $\abs{x_{k'+1}'} \le \ell$, the length of 
$x_i'$ is bounded by $\ell$.
Now, choose $\del_{i+1}$ such that $x_{i+1}' = \del_{i+1} x_{i}'$;
thus, $u_{i+1} \del_{i+1} = v_i$.

For $i = 2,\ldots,k'$
we replace $u_i$ by a shortest word from $[u_i]_L$.
Note that this does not change the fact that all rules respect $L$ (Lemma~\ref{lem:rule:sim}).
We also replace the prefix of $x_i$ and $v_{i-1}$ by this factor.
(There is no need to change $v_{k'}$ as 
$\abs{v_{k'}} = \abs{\del_{k'+1}} \le \abs{x_{k+1}} \le\ell$.)
Therefore, $\abs {x_i} < \abs{x_i'}+m< \ell +m$ and
$r_i \in \Sig^{<m}\times\Sig^{<m}\times\Sig^{< \ell+m}$
if $i\neq 1$ (Lemma~\ref{lem:smallest:words}).
We do not change $u_1$ yet as this may affect the splicing
$(w_0,w_0') \splice s z_1$ if it exists.
Note that, for $i = 2,\ldots,k'$,
we have actually proven a stronger bound than claimed
in statement~2 of Lemma~\ref{lem:series}.
Even though we have not proven the bound for $r_1$ yet,
we have already established
$r_1 \in \Sig^{<m}\times \Sig^*\times\Sig^{<\ell+m}$.
Symmetrically, we can consider statement~2 to be proven
for $i = k'+2,\ldots,k$,
\ie only the prefix $x_1$ and the suffix
$y_1 = y_{k'+1}$ have not been modified yet.

Now, let $x_1 = u_1x_1'$ (as above) and, symmetrically,
let $y_1 = y_{k'+1}'u_{k'+1}$ where $u_{k'+1}$
is the left site of $r_{k'+1}$.
If $k' = 0$ (or $k'=k$), then $u_1$ (\resp $u_{k'+1}$)
can be considered empty and $x_1' = x_{k+1}$
(\resp $y_{k'+1}' = y_{k+1}$).
If $z_1\in I$ we replace $u_1$ and $u_{k'+1}$ by
shortest words from their syntactic classes, respectively,
and the claim holds.
Otherwise, $(w_0,w_0') \splice s z_1$
where $s = (u_0, u_0',v)$, $w_0 = x u_0$,
and $w_0' = u_0' y$, by Lemma~\ref{lem:shorten:sites}.
Thus,
\[
	z_1 = u_1x_1'z y_{k'+1}' u_{k'+1} = x v y.
\]

In the case when $v$ does not overlap with the prefix $u_1$ of $z_1$,
replace $u_1$ by a shortest word from its syntactic class.
If $v$ and the prefix $u_1$ overlap,
let $u_1 = \del_1\del_2$ such that $\del_2$ is the overlap
and replace $\del_1$ and $\del_2$ by a
shortest word from their syntactic classes, respectively.
In both cases, $\abs{u_1} < 2m$ (Lemma~\ref{lem:smallest:words}) and
if $v$ was modified, it got shorter;
hence, we still have $v \in W_\mu$.
Observe that $\abs{x_1} < \ell+2m$ and
$r_1 \in \Sig^{<m}\times \Sig^{<2m}\times\Sig^{< \ell+m}$.
Analogously, $u_{k'+1}$ and $r_{k'+1}$ can be treated in order to conclude the prove of statement~2.
\end{proof}

\subsection{Proof of Theorem~\ref{thm:main}}\label{sec:proof}

Let $L$ be a splicing language and $m = \abs{M_L}$.
Throughout this section, by $\sim$ we denote the equivalence
relation $\sim_L$
and by $[\,\cdot\,]$ we denote the corresponding equivalence classes
$[\,\cdot\,]_L$.

Recall that Theorem~\ref{thm:main} claims that the
splicing system $(I,R)$ with $I = \Sig^{<m^2+ 6m}\cap L$ and
\[
	R = \set{r\in \Sig^{< 2m}\times
	\Sig^{< 2m}\times\Sig^{< m^2+10m}}
	{r\text{ respects }L}
\]
generates $L$.
The proof is divided in two parts.
In the first part, Lemma~\ref{lem:long:words}, we prove that $L$
is generated by a splicing system $(I,R')$ where
all sites of rules in $R'$ are shorter than $2m$,
but we do not care about the lengths of the bridges.
The second part will then conclude the proof by showing that there are no
rules in $R'$ with bridges of length greater than or equal to  $m^2 + 10m$ which are essential
for the creation of the language $L$ by splicing.

\begin{lemma}\label{lem:long:words}
	Let $L$, $m$, and $I$ as above.
	There is $n\in \N$ such that the
	splicing system $(I,R')$ with
	\[
		R' = \set{r\in\Sig^{<2m}\times\Sig^{< 2m}\times\Sig^{\le n}}
		{r\text{ respects }L}
	\]
	generates the language $L = L(I,R')$.
\end{lemma}

\begin{proof}
	As $I \sse L$ and every rule in $R'$ respects $L$,
	it is clear that $L(I,R') \sse L$ for any $n$;
	we only need to prove the converse inclusion.

	As $L$ is a splicing language,
	$L = L(J,S)$ for some splicing system $(J,S)$.
	Let $n$ be larger than the length of every bridge of every rule in $S$
	and $n \ge 4m^2$.
	
	In order to prove $L \sse L(I,R')$
	we use induction on the length of words in $L$.
	For all $w\in L$ with $\abs w < m^2+6m$,
	by definition, $w\in I \sse L(I,R')$.

	Now, consider $w\in L$ with $\abs w \ge m^2+6m$.
	The induction hypothesis states that
	every word $w'\in L$ with $\abs{w'} < \abs {w}$ belongs to $L(I,R')$.
	Factorize $w = x\alp\bet\gam\del y$ such that
	$\abs{x}, \abs{y} = 3m$,
	$\abs{\alp\bet\gam} = m^2$, $\abs\bet \ge 1$,
	$\alp \sim \alp\bet$, and $\gam \sim \bet\gam$.

	The proof idea is to use a pumping argument on $\abg$ in order to obtain
	a very long word.
	This word has to be created by a series of splicings in $(J,S)$.
	We show that these splicings can be
	modified in order to create $w$ by splicing from
	a set of strictly shorter words and with rules from $R'$.
	Then, the induction hypothesis implies $w\in L(I,R')$.

	Choose $j$ sufficiently large ($j > n$ and 
	$J$ does not contain words of length $j$ or more).
	We let $z = \abjgd$ and
	investigate the creation of $xzy \in L$.
	As $z$ is not a factor of a words in $J$, every word in $L$ which contains $z$
	is created by some splicing in $(J,S)$.
	Thus, we can trace back the creation of $xzy$ by splicing to the point
	where the factor $z$ is affected for the last time.
	Let $z_{k+1} = x_{k+1}zy_{k+1}$, where $x_{k+1} = x$ and $y_{k+1} = y$,
	be created by $k$ splicings from a word $z_1 = x_1 z y_1$
	where $x_1zy_1$ is created by a splicing
	$(w_0,w_0')\splice{s} z_1$ with $w_0, w_0'\in L$,
	$s\in S$, and the bridge of $s$ overlaps with $z$ in $z_1$.
	Furthermore, for $i = 1,\ldots,k$ the intermediate splicings are either
	\begin{enumerate}[(i)]
		\item $(w_i,z_i)\splice{r_i} x_{i+1}zy_{i+1}=z_{i+1}$,
			$w_i\in L$, $r_i\in S$, $y_{i+1} = y_i$,
			and the bridge of $r_i$
			is covered by the prefix $x_{i+1}$ or
		\item $(z_i,w_i)\splice{r_i} x_{i+1}zy_{i+1}=z_{i+1}$,
			$w_i\in L$, $r_i\in S$, $x_{i+1} = x_i$,
			and the bridge of $r_i$
			is covered by the suffix $y_{i+1}$.
	\end{enumerate}
	Following Lemma~\ref{lem:series} (with $\ell = 3m$),
	we may assume that $w_1,\ldots,w_k \in I$,
	$r_1,\ldots,r_k \in \Sig^{<2m}\times\Sig^{<2m}\times \Sig^{<4m}$,
	thus $r_1,\ldots,r_k\in R'$,
	and $\abs{x_1},\abs{y_1} < 5m$.
	Furthermore, we may use the same words and 
	rules in order to create $w = x_{k+1}\abgd y_{k+1}$
	from $x_1\abgd y_1$ by splicing, \ie
	if $x_1\abgd y_1$ belongs to $L(I,R')$, so does $w$.
	
	Now, consider the first splicing
	$(w_0,w_0')\splice s z_1 = x_1zy_1$.
	By Lemma~\ref{lem:shorten:sites},
	we assume $s = (u_1,u_2;v)$ such that
	$w_0 = xu_1$, $w_0' = u_2 y$
	and $\abs{u_1},\abs{u_2} <m$
	($x$ and $y$ are newly chosen words).
	Hence,
	\[
		z_0 = xvy = x_1\abjgd y_1.
	\]
	where $x$ is a proper prefix of $x_1\abjgd$ and
	$y$ is a proper suffix of $\abjgd y_1$.
	
	We will now pump down the factor $\bet^j$ to $\bet$
	in order to obtain the words $\ti x$, $\ti v$, $\ti y$
	from $x$, $v$, $y$, respectively, as follows:

\begin{asparaenum}
	\item
	If $v$ overlaps with $\bet^j$ but does neither cover $\alp$
	nor $\gam$, extend $v$ (Lemma~\ref{lem:rule:ext})
	such that $v = \alp\bet^j\gam$.
	Observe that, now, the factor $\alp\bet^j\gam$ is 
	covered by either $xv$ or $vy$.
	
	\item
	If $\alp\bet^j$ or $\bet^j\gam$ is covered by one of
	$x$, $v$, or $y$, then replace this factor
	by $\alp\bet$ or $\bet\gam$, respectively.
	Otherwise, by symmetry, assume that
	$\alp\bet^j\gam$ is covered by $xv$ and, therefore,
	we can factorize
	\begin{align*}
		x &= x_1\alp\bet^{j_1}\bet_1 &
		v &= \bet_2\bet^{j_2}\gam v'
	\end{align*}
	where $\bet_1\bet_2 = \bet$ and $j_1+j_2+1 = j$.
	The results of pumping are the words $\ti x = x_1\alp\bet_1$,
	$\ti v = \bet_2\gam v'$, and $\ti y = y$.
\end{asparaenum}
	
	Let $\ti u_1$ and $\ti u_2$ be the sites of $s$ that
	may have been altered due to the extension of $v$ and,
	by Lemma~\ref{lem:shorten:sites},
	assume $\abs{\ti u_1},\abs{\ti u_2} < m$.	
	If we used an extension for $v$, then $\abs{\ti v} = m^2$.
	No matter whether we used an extension, $t = (\ti u_1,\ti u_2;\ti v)\in R'$
	and $(\ti x\ti u_1,\ti u_2\ti y)\splice{t} x_1 \abgd y_1$ as desired.
	Observe that $\ti x$ is a prefix of $x_1\abgd$ and $\ti y$ is  a
	suffix of $\abgd y_1$ and recall that $\abs{x_1},\abs{y_1} <5m$.
	Therefore, $\abs{\ti x\ti u_1},\abs{\ti u_2\ti y} < \abs{\abgd}+6m =\abs w$
	and, by induction hypothesis,
	$\ti x\ti u_1$ and $\ti u_2\ti y$ belong to $L(I,R')$.
	We conclude that $x_1\abgd y_1$ as well as $w$ belong to $L(I,R')$.
\end{proof}

We are now prepared to prove the main result.

\begin{proof}[Proof of Theorem~\ref{thm:main}]
	Recall that for a splicing language $L$ with $m = \abs{M_L}$,
	we intend to prove that the splicing system
	$(I,R)$ with $I = \Sig^{< m^2+6m}\cap L$ and
	\[
		R = \set{r\in\Sig^{<2 m}\times\Sig^{< 2m}\times\Sig^{< m^2+10m}}
		{r\text{ respects }L}
	\]
	generates the language $L = L(I,R)$.
	Obviously, $L(I,R) \subseteq L$.
	By Lemma~\ref{lem:long:words}, there is a finite set of rules
	$R'\sse \Sig^{<2m}\times \Sig^{<2m} \times \Sig^*$
	such that $L(I,R') = L$.

	For a word $\mu$ we let $W_\mu = \set{w\in\Sig^*}{w\eLL \mu}$,
	as we did before.
	Define the set of rules where every component
	is length-lexicographically bounded by $\mu$
	\[
		R_\mu = \set{r\in\Sig^{< 2m}\times\Sig^{< 2m}\times W_\mu}
		{r\text{ respects }L}
	\]
	and the language $L_\mu = L(I,R_\mu)$; clearly, $L_\mu \sse L$.
	For two words $\mu \eLL v$ we see that $R_\mu\sse R_v$,
	and hence, $L_\mu\sse L_v$.	
	Thus, if $L_\mu = L$ for some word $\mu$, then
	for all words $v$ with $\mu \eLL v$, we have $L_v = L$.
	As $L = L(I,R')$, there exists a word $\mu$ such that $L_\mu = L$.
	Let $\mu$ be the smallest word, in the length-lexicographic order, such that
	$L_\mu = L$.
	Note that if $\abs\mu < m^2+10$, then $R_\mu \sse R$ and
	$L = L_\mu \sse L(I,R)$.
	For the sake of contradiction assume $\abs \mu \ge m^2+10m$.
	Let $\nu$ be the next-smaller word than $\mu$,
	in the length-lexicographic order,
	and let $S = R_{\nu}$.
	Note that $L(I,S) \ssne L$ and $R_\mu\sm S$ contains only rules
	whose bridges are $\mu$.
	
	Choose $w$ from $L\sm L(I,S)$ as a shortest word, \ie
	for all $w'\in L$ with $\abs {w'} < \abs w$, we have $w' \in L(I,S)$.
	Factorize $w = xzy$ with $\abs x = \abs y = 3m$;
	note that $\abs w \ge m^2 + 6m$ since, otherwise, $w\in I$.
	Factorize $\mu = \del_1\abg\del_2$ with
	$\abs{\del_1},\abs{\del_2} \ge 5m$,
	$\abs{\abg} = m^2$, $\bet\neq \e$, $\alp\sim \alp\bet$,
	and $\gam \sim \bet\gam$, by Lemma~\ref{lem:pumping}.
	
	Next, we will use a pumping argument on all factors $\abg$ in $z$.
	As in the proof of Lemma~\ref{lem:long:words}, this new word
	has to be created by a series of splicings in $(I,R_\mu)$
	and we will show that these splicings can be modified in order to
	create $w$ from strictly shorter words and with rules from $S$.
	This will contradict the assumption that $w$ is a shortest
	word from $L\sm L(I,S)$.
	
	Let $j$ be a sufficiently large even number ($j > 4\abs\mu+\abs z$ will do).
	We define a word $\ti z$ which is the result of applying the
	pumping algorithm from Lemma~\ref{lem:pumping:algorithm} on $z$,
	as discussed in Section~\ref{sec:pump:alg}.
	The pumping algorithm replaces the occurrences of $\abg$ in $z$ by $\abjg$
	such that for every factor $\ti z\pos{k,k+m^2} = \abg$, either
	\begin{enumerate}[(a)]
		\item $\alpha\beta^{j\slash 2}$ is a factor of $\ti z$ starting at position
			$\ti z\pos{k}$ or 
		\item $\beta^{j\slash2}\gamma$ is a factor of
			$\ti z$ ending at position $\ti z\pos{k+m^2}$
	\end{enumerate}
	holds.
	In particular, if $\del_1\abg\del_2$ is a factor of $\ti z$ either
	\A $\gamma\del_2$ is a prefix of a word in $\bet^+$ or
	\B $\delta_1\alpha$ is a suffix of a word in $\bet^+$.
	By induction and as $\abg \sim \abjg$, it is easy to see that
	$z\sim \ti z$ and $x\ti zy\in L$.
	
	Let us trace back the creation of $x\ti zy\in L$
	by splicing in $(I,R_\mu)$
	to a word $x_1\ti zy_1$
	where either $x_1\ti zy_1\in I$
	or where $x_1\ti zy_1$ is created by a splicing
	that affects $\ti z$.
	Let $z_{k+1} = x_{k+1}\ti zy_{k+1}$, where $x_{k+1} = x$ and $y_{k+1} = y$,
	be created by $k$ splicings from a word $z_1 = x_1\ti z y_1$
	where either $x_1\ti z y_1\in I$ or $x_1\ti zy_1$ is created
	by a splicing $(w_0,w_0')\splice{s} z_1$ with $w_0,w_0'\in L$,
	$s\in R_\mu$, and the bridge of $s$ overlaps with $\ti z$.
	Furthermore, for $i = 1,\ldots,k$ the intermediate splicings are either
	\begin{enumerate}[(i)]
		\item $(w_i,z_i)\splice{r_i} x_{i+1}\ti zy_{i+1}=z_{i+1}$,
			$w_i\in L$, $r_i\in R_\mu$, $y_{i+1} = y_i$,
			and the bridge of $r_i$
			is covered by the prefix $x_{i+1}$ or
		\item $(z_i,w_i)\splice{r_i} x_{i+1}\ti zy_{i+1}=z_{i+1}$,
			$w_i\in L$, $r_i\in R_\mu$, $x_{i+1} = x_i$,
			and the bridge of $r_i$
			is covered by the suffix $y_{i+1}$.
	\end{enumerate}
	Following Lemma~\ref{lem:series} (with $\ell = 3m$), we may assume that
	$w_1,\ldots,w_k \in I$,
	$r_1,\ldots,r_k \in \Sig^{<2m}\times\Sig^{<2m}\times \Sig^{<4m}$,
	thus $r_1,\ldots,r_k\in S$,
	and $\abs{x_1},\abs{y_1} < 5m$.
	Furthermore, we may use the same words and 
	rules in order to create $w = x_{k+1}z y_{k+1}$
	from $x_1 z y_1$ by splicing.
	As $w$ does not belong to $L(I,S)$,
	the word $x_1z y_1$ must not belong to $L(I,S)$ either.
	If $z_1$ was in $I$, then $x_1zy_1\in I$ as well, as $z$ is 
	at most as long as $\ti z$.

	Therefore, $z_1$ is created by a splicing
	$(w_0,w_0')\splice s z_1$ where 
	$s = (u_1,u_2;v)$, $w_0 = x  u_1$, and $w_0' =  u_2 y$
	where $\abs{u_1},\abs{u_2} < m$,
	by Lemma~\ref{lem:shorten:sites}
	(here, $x$ and $y$ are newly chosen words).
	We have
	\[
		z_1 = x_1\ti zy_1 = xvy
	\]
	where $x$ is a proper prefix of $x_1\ti z$
	and $y$ is a proper suffix of $\ti z y_1$.
	Recall that either $s\in S$ or $v = \mu$.

	However, we will see next that if $ v = \mu$,
	there is also a rule $\ti s\in S$ and slightly modified words
	which can be used in order to create $x_1\ti z y_1$ by splicing.
	In this case $\mu = \del_1\abg\del_2$ is a factor of $z_1$.
	As $\abs{\del_1},\abs{\del_2} \ge 5m>\abs{x_1},\abs{y_1}$,
	the factor $\abg$ is covered by $\ti z$ and, as such,
	the pumping algorithm ensured that
	either \A $\alp$ is succeeded by $\bet^{j/2}$ or
	\B $\gam$ is preceded by $\bet^{j/2}$.
	Due to symmetry, we only consider the former case,
	in which $\gam \del_2$ is a prefix of
	a word in $\bet^+$.
	Let us shorten the bridge $v$ such that
	$\ti s = (u_1,u_2; \del_1\alp\gam\del_2)$.
	Note that $\ti s\in S$
	(as $\alp\sim \alp\bet$ and by Lemma~\ref{lem:rule:sim}).
	Furthermore, as $j$ is large enough, $y= \bet_2\bet^{\ell}y'$
	where $\bet_2$ is the suffix of $\bet$
	such that $\gam\del_2\bet_2 \in\bet^+$ and $\ell \ge \abs{\gam}$.
	Note that this implies $\bet_2\gam$ is a prefix of $y$,
	which allows us to add an additional $\bet$.
	Therefore,
	$(w_0, u_2\bet_2\bet^{\ell+1}y') \splice{\ti s} z_1$
	where $u_2\bet_2\bet^{\ell+1}y'\in L$.
	This observation justifies the assumption that
	$v \neq \mu$ and $s\in S$
	which we will make for the remainder of the proof.

	Next, we will pump down the factors $\alp\bet^j\gam$ to $\abg$
	in $\ti z$ again.
	At every position where we pumped up before, we are now pumping down
	(in reverse order)
	in order to obtain the words $\ti x$, $\ti v$, $\ti y$ from
	the words $x$, $v$, $y$, respectively.
	The pumping in each step
	is done as in the proof of Lemma~\ref{lem:long:words}:

\begin{asparaenum}
	\item
	If $v$ overlaps with $\bet^j$ (in the factor that we are pumping down)
	but it neither covers $\alp$ nor $\gam$, extend $v$ (Lemma~\ref{lem:rule:ext})
	such that $v = \alp\bet^j\gam$.
	Observe that, now, the factor $\alp\bet^j\gam$ is 
	covered by either $xv$ or $vy$.
	
	\item
	If $\alp\bet^j$ or $\bet^j\gam$ is covered by one of
	$x$, $v$, or $y$, then replace this factor
	by $\alp\bet$ or $\bet\gam$, respectively.
	Otherwise, by symmetry, assume that
	$\alp\bet^j\gam$ is covered by $xv$ and, therefore,
	we can factorize
	\begin{align*}
		x &= x'\alp\bet^{j_1}\bet_1 &
		v &= \bet_2\bet^{j_2}\gam v'
	\end{align*}
	where $\bet_1\bet_2 = \bet$ and $j_1+j_2+1 = j$.
	The results of pumping are the words $x'\alp\bet_1$,
	$\bet_2\gam v'$.
\end{asparaenum}

	Let $\ti u_1$ and $\ti u_2$ be the sites of $s$ that
	may have been altered due to extensions and, by Lemma~\ref{lem:shorten:sites},
	assume $\abs{\ti u_1},\abs{\ti u_2} < m$.	
	If we used an extension for $v$ in at least one of the steps,
	then $\abs{\ti v} \le m^2$.
	No matter whether or not we used an extension,
	$t = (\ti u_1,\ti u_2;\ti v)\in S$
	and $(\ti x\ti u_1,\ti u_2\ti y)\splice{t} x_1 z y_1$.
	As $\abs{\ti x\ti u_1},\abs{\ti u_2\ti y} < \abs{z}+6m =\abs w$,
	$\ti x\ti u_1$ and $\ti u_2\ti y$ belong to $L(I,S)$.
	We conclude that $x_1 z y_1$ as well as $w$ belong to $L(I,S)$ ---
	the desired contradiction.
\end{proof}

\section{The Case of Classical Splicing}
\label{sec:main}

In this section, we consider the splicing operation as defined in \cite{Paun96}.
This is the most commonly used definition for splicing in formal language
theory.
The notation we use has been employed in previous papers, see \eg \cite{BonizzoniFZ05,GoodeP07}.
Throughout this section,
a quadruplet of words $r = (u_1,v_1;u_2,v_2)\in (\Sig^*)^4$ is called a
{\em (splicing) rule}.
The words $u_1v_1$ and $u_2v_2$ are called {\em left}
and {\em right site} of $r$, respectively.
This splicing rule can be applied to two words 
$w_1 = x_1u_1v_1y_1$ and $w_2 = x_2u_2v_2y_2$,
that each contain one of the sites,
in order to create the new word $z = x_1u_1v_2y_2$,
see Figure~\ref{fig:splice}.
This operation is called {\em splicing} and
it is denoted by $(w_1,w_2)\splice{r} z$.
The {\em splicing position} of this splicing is $z\pos{\abs{x_1u_1}}$; that is the position between the factors $x_1u_1$ and $v_2y_2$ in $z$.

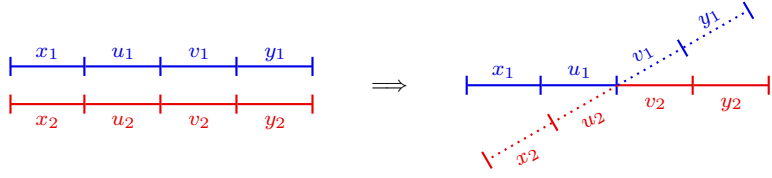
\begin{figure}[ht]
	\centering
	\begin{tikzpicture}[text height=1ex,text depth=0ex,
			font={\footnotesize}]
		\begin{scope}[yshift=.25cm,thick]
			\draw [|-|,eins] (0,0) -- node [above] {$x_1$} (1,0);
			\draw [-|,eins] (1,0) -- node [above] {$u_1$} (2,0);
			\draw [-|,eins] (2,0) -- node [above] {$v_1$} (3,0);
			\draw [-|,eins] (3,0) -- node [above] {$y_1$} (4,0);

			\draw [|-|,zwei] (0,-.5) -- node [below] {$x_2$} (1,-.5);
			\draw [-|,zwei] (1,-.5) -- node [below] {$u_2$} (2,-.5);
			\draw [-|,zwei] (2,-.5) -- node [below] {$v_2$} (3,-.5);
			\draw [-|,zwei] (3,-.5) -- node [below] {$y_2$} (4,-.5);
		\end{scope}

		\node at (5,0) {$\Longrightarrow$};

		\begin{scope}[xshift=6cm,thick]
			\draw [|-|,eins] (0,0) -- node [above] {$x_1$} (1,0);
			\draw [-|,eins] (1,0) -- node [above] {$u_1$} (2,0);
			\draw [-|,zwei] (2,0) -- node [below] {$v_2$} (3,0);
			\draw [-|,zwei] (3,0) -- node [below] {$y_2$} (4,0);
			
			\def\ang{30}

			\draw [-|,dotted,eins] (2,0) -- 
				node [above,sloped] {$v_1$} +(\ang:1);
			\draw [-|,dotted,eins] (2,0)++(\ang:1) --
				node [above,sloped] {$y_1$} ++(\ang:1);
				
			\draw [-|,dotted,zwei] (2,0) -- 
				node [below,sloped] {$u_2$} +(180+\ang:1);
			\draw [-|,dotted,zwei] (2,0)++(180+\ang:1) --
				node [below,sloped] {$x_2$} ++(180+\ang:1);
		\end{scope}
	\end{tikzpicture}
	\caption{Splicing of the words $x_1u_1v_1y_1$ and $x_2u_2v_2y_2$
		by the rule $r = (u_1,v_1;u_2,v_2)$.}
	\label{fig:splice}
\end{figure}

Just as in Section~\ref{sec:pixton},
for a rule $r$ we define the {\em splicing operator} $\sig_r$
such that for a language $L$
\[
	\sig_r(L) = \set{z\in\Sig^*}
		{\exists w_1,w_2\in L \colon (w_1,w_2)\splice{r} z}
\]
and for a set of splicing rules $R$, we let
\[
	\sig_R(L) = \bigcup_{r\in R} \sig_r(L).
\]
The reflexive and transitive closure of the splicing operator $\sig_R^*$
is given by
\begin{align*}
	\sig_R^0(L) &= L, &
	\sig_R^{i+1}(L) &= \sig_R^i(L) \cup \sig_R(\sig_R^i(L)), &
	\sig_R^*(L) &= \bigcup_{i\ge 0} \sig_R^i(L).
\end{align*}
A finite set of axioms $I \sse \Sig^*$ and a
finite set of splicing rules $R \sse (\Sig^*)^4$
form a {\em splicing system} $(I,R)$.
Every splicing system $(I,R)$ generates a language
$L(I,R) = \sig_R^*(I)$.
Note that $L(I,R)$ is the smallest language which is closed under
the splicing operator $\sig_R$ and includes $I$.
It is known that the language generated by a splicing system
is regular, see \cite{CulikH91,Pixton96}.
A (regular) language $L$ is called a {\em splicing language}
if a splicing system $(I,R)$ exists such that $L = L(I,R)$.

A rule $r$ is said to {\em respect} a language $L$
if $\sig_r(L) \sse L$.
It is easy to see that for any splicing system $(I,R)$,
every rule $r \in R$ respects the generated language $L(I,R)$.
Moreover, a rule $r\notin R$ respects $L(I,R)$ if and only if
\mbox{$L(I,R\cup\oneset r)={}$}\mbox{$L(I,R)$}.
We say a splicing $(w_1,w_2)\splice r z$ {\em respects} 
a language $L$ if  $w_1,w_2\in L$ and $r$ respects $L$;
obviously, this implies $z\in L$, too.

The main result of this section states that, if a regular language $L$ is a splicing language, then it is generated by a particular splicing system $(I,R)$ which only depends on the syntactic monoid of $L$.

\begin{theorem}\label{thm:main:four}
	Let $L$ be a splicing language and $m = \abs{M_L}$.
	The splicing system $(I,R)$ with $I = \Sig^{<m^2+ 6m}\cap L$ and
	\[
		R = \set{r\in\Sig^{< m^2+10m}\times\Sig^{< 2m}\times
		\Sig^{< 2m}\times\Sig^{< m^2+10m}}
		{r\text{ respects }L}
	\]
	generates the language $L = L(I,R)$.
\end{theorem}

As the language generated by the splicing system $(I,R)$ is constructible,
Theorem~\ref{thm:main:four} implies that the problem whether or not a given
regular language is a splicing language is decidable.
A detailed discussion of the decidability result 
is given in Section~\ref{sec:decidability}.

Let $L$ be a formal language.
Clearly, every set of words $J\sse L$ and set of rules $S$ where every rule in $S$ respects $L$ generates a subset $L(J,S)\sse L$.
Therefore, in Theorem~\ref{thm:main:four} the inclusion $L(I,R) \sse L$ is obvious.
The rest of this section is devoted to the proof of the converse inclusion $L\sse L(I,R)$.
The proof uses many ideas that have been employed in the Section~\ref{sec:pixton}.
However, there are some challenges we encounter solely while considering the classic
splicing variant.
The additional complexity comes from having to 
handle the first and fourth components of rules,
which in the case of classical splicing occur both in the words used for splicing and the splicing result.
In contrast, in Pixton splicing the sites of a rule do not occur in splicing result, whereas the bridge is not a factor of the words used for splicing.
The structure of this section is the same as Section~\ref{sec:pixton}.
In Section~\ref{sec:rules:four} we will present techniques to obtain rules that respect a regular language $L$ from other rules that respect $L$, and we show how we can modify a splicing step, such that the words used for splicing are not significantly longer than the splicing result;
similar results can be found in \cite{GoodePHD,GoodeP07}.
In Section~\ref{sec:series:four} we use these techniques to
show that a long word $z\in L$ can be obtained by a series of splicings 
from a set shorter words from $L$ and by using rules which satisfy certain
length restrictions.
Finally, in Section~\ref{sec:proof:four} we prove Theorem~\ref{thm:main:four}.

\subsection{Rule Modifications}\label{sec:rules:four}

The first lemma states us that we can extend the sites of a rule $r$ 
such that the extended rule respects all languages that are respected by $r$.

\begin{lemma}\label{lem:extension}
	Let $r = (u_1,v_1;u_2,v_2)$ be a rule which respects a language $L$.
	For every word $x$,
	the rules $(xu_1,v_1;u_2,v_2)$, $(u_1,v_1x;u_2,v_2)$,
	$(u_1,v_1;xu_2,v_2)$, and $(u_1,v_1;u_2,v_2x)$ respect $L$ as well.
\end{lemma}

\begin{proof}
	Let $s$ be any of the rules
	$(xu_1,v_1;u_2,v_2)$, $(u_1,v_1x;u_2,v_2)$,
	$(u_1,v_1;xu_2,v_2)$, $(u_1,v_1;u_2,v_2x)$.
	In order to prove that $s$ respects $L$,
	we have to show that,
	for all $w_1,w_2\in L$ and $z\in\Sig^*$
	such that $(w_1,w_2) \splice s z$, we have $z\in L$, too.
	Indeed, if $(w_1,w_2) \splice s z$, then
	$(w_1,w_2) \splice r z$ and, as $r$ respects $L$,
	we conclude $z\in L$.
\end{proof}

Henceforth, for a rule $r = (u_1,v_1;u_2,v_2)$, we will refer to the rules $(xu_1,v_1;u_2,v_2)$ and  $(u_1,v_1x;u_2,v_2)$ as extensions of the left site of $r$ and to $(u_1,v_1;xu_2,v_2)$ and $(u_1,v_1;u_2,v_2x)$ as extensions of the right site of $r$.

Next, for a language $L$,
let us investigate the syntactic class of a rule $r = (u_1,v_1;u_2,v_2)$.
The {\em syntactic class} (with respect to $L$) of $r$ is the set
of rules $[r]_L = [u_1]_L \times [v_1]_L\times [u_2]_L\times [v_2]_L$ and 
two rules $r$ and $s$ are {\em syntactically congruent} (with respect to $L$),
denoted by $r \sim_L s$, if $s \in [r]_L$.

\begin{lemma}\label{lem:congruence}
	Let $r$ be a rule which respects a language $L$.
	Every rule $s\in[r]_L$ respects $L$.
\end{lemma}

\begin{proof}
	Let $r = (u_1,v_1; u_2,v_2)$ and $s = (\ti u_1,\ti v_1; \ti u_2,\ti v_2)$.
	Thus, $u_i\sim_L \ti u_i$ and $v_i\sim\ti v_i$ for $i= 1,2$.
	For $\ti w_1 = x_1\ti u_1\ti v_1y_1\in L$ and
	$\ti w_2 = x_2\ti u_2\ti v_2y_2\in L$
	we have to show that $\ti z = x_1\ti u_1\ti v_2 y_2\in L$.
	For $i = 1,2$, let $w_i = x_iu_iv_iy_i$ and note that $w_i \sim_L \ti w_i$;
	hence, $w_i\in L$.
	Furthermore, $(w_1,w_2)\splice r x_1 u_1v_2y_2 = z \in L$
	as $r$ respects $L$,
	and $\ti z\in L$ as $z\sim_L \ti z$.
\end{proof}

Consider a splicing $(x_1u_1v_1y_1,x_2u_2v_2y_2)\splice r x_1u_1v_2y_2$
which respects a regular language $L$,
as shown in Figure~\ref{fig:shorten} on the left site.
The factors $v_1y_1$ and $x_2u_2$ may be relatively long but they do
not occur as factors in the resulting word $x_1u_1v_2y_2$.
In particular, it is possible that two long words are spliced 
and the outcome is a relatively short word.
Using the Lemmas~\ref{lem:extension} and~\ref{lem:congruence},
we can find shorter words in $L$ and a modified splicing rule
which can be used to obtain $x_1u_1v_2y_2$.

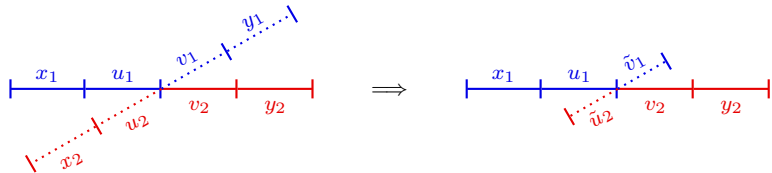
\begin{figure}[ht]
	\center
	\begin{tikzpicture}[text height=1ex,text depth=0ex,
			font={\footnotesize}]
		\begin{scope}[thick]
			\draw [|-|,eins] (0,0) -- node [above] {$x_1$} (1,0);
			\draw [-|,eins] (1,0) -- node [above] {$u_1$} (2,0);
			\draw [-|,zwei] (2,0) -- node [below] {$v_2$} (3,0);
			\draw [-|,zwei] (3,0) -- node [below] {$y_2$} (4,0);
			
			\def\ang{30}

			\draw [-|,dotted,eins] (2,0) -- 
				node [above,sloped] {$v_1$} +(\ang:1);
			\draw [-|,dotted,eins] (2,0)++(\ang:1) --
				node [above,sloped] {$y_1$} ++(\ang:1);
				
			\draw [-|,dotted,zwei] (2,0) -- 
				node [below,sloped] {$u_2$} +(180+\ang:1);
			\draw [-|,dotted,zwei] (2,0)++(180+\ang:1) --
				node [below,sloped] {$x_2$} ++(180+\ang:1);
		\end{scope}

		\node at (5,0) {$\Longrightarrow$};

		\begin{scope}[xshift=6cm,thick]
			\draw [|-|,eins] (0,0) -- node [above] {$x_1$} (1,0);
			\draw [-|,eins] (1,0) -- node [above] {$u_1$} (2,0);
			\draw [-|,zwei] (2,0) -- node [below] {$v_2$} (3,0);
			\draw [-|,zwei] (3,0) -- node [below] {$y_2$} (4,0);
			
			\def\ang{30}

			\draw [-|,dotted,eins] (2,0) -- 
				node [above,sloped] {$\ti v_1$} +(\ang:.75);
				
			\draw [-|,dotted,zwei] (2,0) -- 
				node [below,sloped] {$\ti u_2$} +(180+\ang:.75);
		\end{scope}
	\end{tikzpicture}
	\caption{Replacing $v_1y_1$ and $x_2u_2$ by {\em short} words.}
	\label{fig:shorten}
\end{figure}

\begin{lemma}\label{lem:shorten}
	Let $r = (u_1,v_1;u_2,v_2)$ be a rule which respects a regular language $L$
	and $w_1 = x_1u_1v_1y_1\in L$, $w_2 = x_2u_2v_2y_2\in L$.
	There is a rule $s = (u_1,\ti v_1;\ti u_2,v_2)$ which respects $L$
	and words $\ti w_1 = x_1 u_1\ti v_1\in L$, $\ti w_2 = \ti u_2v_2 y_2\in L$
	such that $\abs{\ti v_1}, \abs{\ti u_2} < \abs{M_L}$.
	More precisely, $\ti v_1\in [v_1y_1]_L$ and $\ti u_2\in [x_2u_2]_L$.
	
	In particular, whenever $(w_1,w_2)\splice r x_1u_1v_2y_2 = z$,
	then there is a splicing $(\ti w_1,\ti w_2)\splice s z$
	which respects $L$ where $\ti w_1$, $\ti w_2$, and $s$ have the properties
	described above.
\end{lemma}

\begin{proof}
	By Lemma~\ref{lem:extension}, the rule
	$(u_1,v_1y_1;x_2u_2,v_2)$ respects $L$.
	Choose $\ti v_1\in [v_1y_1]_L$ and $\ti u_2\in [x_2u_2]_L$
	as shortest words from the sets, respectively.
	By Lemma~\ref{lem:smallest:words},
	$\abs{\ti u_1}, \abs{\ti u_2} < \abs{M_L}$ and
	$\ti w_1 = x_1u_1\ti v_1\in L$, $\ti w_2 = \ti u_2v_2y_2\in L$.
	Furthermore, by Lemma~\ref{lem:congruence},
	$s = (u_1,\ti v_1;\ti u_2,v_2)$ respects $L$.
\end{proof}

\subsection{Series of Splicings}\label{sec:series:four}

Consider the creation of words by a series of splicings.
Let us begin with a simple observation.
In the case when a word is created by two (or more) successive splicings, but none of the splicing sites overlaps the position of the other splicing, the order of these splicings is irrelevant.
Recall that the splicing position of a splicing $(w_1,w_2)\splice r z$ with
$r = (u_1,v_1;u_2,v_2)$ is the position between the factors $u_1$ and $v_2$ in $z$.
The notation in Remark~\ref{rem:swap} is the same as in the Figure~\ref{fig:swap}.

\begin{figure}[ht]
	\center
	\begin{tikzpicture}[text height=1ex,text depth=0ex,
			font={\footnotesize}]
			
		\def\ang{30}
		\begin{scope}[thick]
			\draw [|-|,eins] (0,0) -- node [above] {$x_1$} (1,0);
			\draw [-|,eins] (1,0) -- node [above] {$u_1$} (2,0);
			\draw [-|,dotted,eins] (2,0) -- 
				node [above,sloped] {$v_1$} +(\ang:1);
			\draw [-|,dotted,eins] (2,0)++(\ang:1) --
				node [above,sloped] {$y_1$} ++(\ang:1);

			\draw [-|,dotted,zwei] (2,0) -- 
				node [below,sloped] {$u_2$} +(180+\ang:1);
			\draw [-|,dotted,zwei] (2,0)++(180+\ang:1) --
				node [below,sloped] {$x_2$} ++(180+\ang:1);
			\draw [-|,zwei] (2,0) --
				node [fill=white,inner sep=1pt] {$z_2$} (5,0);
			\draw [-|,dotted,zwei] (5,0) -- 
				node [above,sloped] {$v_2'$} +(\ang:1);
			\draw [-|,dotted,zwei] (5,0)++(\ang:1) --
				node [above,sloped] {$y_2$} ++(\ang:1);

			\draw [-|,drei] (5,0) -- node [below] {$v_3$} (6,0);
			\draw [-|,drei] (6,0) -- node [below] {$y_3$} (7,0);
			\draw [-|,dotted,drei] (5,0) -- 
				node [below,sloped] {$u_3$} +(180+\ang:1);
			\draw [-|,dotted,drei] (5,0)++(180+\ang:1) --
				node [below,sloped] {$x_3$} ++(180+\ang:1);
		\end{scope}

		\draw [decorate,decoration=brace,zwei] (3,.15) to 
			node [above] {$u_2'$} (5,.15);
		\draw [decorate,decoration=brace,zwei] (4,-.15) to 
			node [below] {$v_2$} (2,-.15);
	\end{tikzpicture}
	\caption{The word $x_1u_1z_2v_3y_3$ can be created either by using the
		right splicing first or by using the left splicing first.}
	\label{fig:swap}
\end{figure}
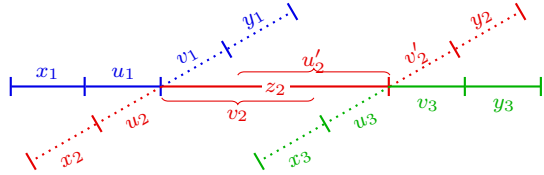

\begin{remark}\label{rem:swap}
	Let $w_1 = x_1u_1v_1y_1$, $w_2 = x_2u_2z_2 v_2' y_2$,
	where $v_2$ is a prefix of $z_2$ and $u_2'$ is a suffix of $z_2$,
	$w_3 = x_3u_3v_3y_3$ be words and
	$r_1 = (u_1,v_1;u_2,v_2)$, $r_2 = (u_2',v_2';u_3,v_3)$ be rules.
	In order to create the word
	$z = x_1u_1z_2 v_3 y_3$ by splicing,
	we may use splicings
	\begin{align*}
		(w_1,w_2)&\splice{r_1}  x_1u_1z_2 v_2' y_2 = z', &
		(z',w_3)&\splice{r_2} z \qquad\text{or} \\
		(w_2,w_3)&\splice{r_2}  x_2u_2z_2 v_3 y_3 = z'', &
		(w_1,z'')&\splice{r_1} z.
	\end{align*}
\end{remark}

Now, consider a word $z$ which is created by two successive splicings
from words $w_i = x_iu_iv_iy_i$ for $i = 1,2,3$ as in Figure~\ref{fig:combine:rules}.
If no factor of $w_1$ is a part of $z$,
then we can find another splicing rule $s$ such that
$(w_3,w_2)\splice s z$.
This replacement will become crucial in the proof of Lemma~\ref{lem:creation}.

\begin{figure}[ht]
	\center
	\begin{tikzpicture}[text height=1ex,text depth=0ex,
			font={\footnotesize}]
		\def\ang{30}

		\begin{scope}[thick]
			\draw [|-|,eins] (0,0) -- node [above] {$x_1$} (.75,0);
			\draw [-|,eins] (.75,0) -- node [above] {$u_1$} (1.5,0);
			\draw [-|,dotted,eins] (1.5,0) -- 
				node [above,sloped] {$v_1$} +(\ang:.75);
			\draw [-|,dotted,eins] (1.5,0)++(\ang:.75) --
				node [above,sloped] {$y_1$} ++(\ang:.75);
			
			\draw [-|,zwei] (1.5,0) -- node [below] {$v_2$} (2.25,0);
			\draw [-|,zwei] (2.25,0) -- node [below] {$y_2$} (3,0);
			\draw [-|,dotted,zwei] (1.5,0) -- 
				node [below,sloped] {$u_2$} +(180+\ang:.75);
			\draw [-|,dotted,zwei] (1.5,0)++(180+\ang:.75) --
				node [below,sloped] {$x_2$} ++(180+\ang:.75);
		\end{scope}

		\node at (3.5,-.25) {$+$};

		\begin{scope}[xshift=4.25cm,thick]
			\draw [|-|,drei] (0,0) -- node [above] {$x_3$} (.75,0);
			\draw [-|,drei] (.75,0) -- node [above] {$u_3$} (1.5,0);
			\draw [-|,dotted,drei] (1.5,0) -- 
				node [above,sloped] {$v_3$} +(\ang:.75);
			\draw [-|,dotted,drei] (1.5,0)++(\ang:.75) --
				node [above,sloped] {$y_3$} ++(\ang:.75);
			
			\draw [-|,zwei] (1.5,0) -- node [below,pos=.25] {$v_2$} (2,0);
			\draw [-|,zwei] (2,0) -- node [below] {$y_2$} (2.75,0);
			\draw [-|,dotted,zwei] (1.5,0) -- +(180+\ang:.25);
			\draw [-|,dotted,eins] (1.5,0)++(180+\ang:.25) --
				node [below,sloped] {$u_1$} ++(180+\ang:.75);
			\draw [-|,dotted,eins] (1.5,0)++(180+\ang:1) --
				node [below,sloped] {$x_1$} ++(180+\ang:.75);
		\end{scope}
		\begin{scope}[xshift=4.25cm]
			\draw [decorate,decoration={brace,mirror}] (1.5,-.4) to 
				node [below] {$v_4y_4$} (2.75,-.4);
		\end{scope}

		\node at (8,-.25) {$\Longrightarrow$};

		\begin{scope}[xshift=9.25cm,thick]
			\draw [|-|,drei] (0,0) -- node [above] {$x_3$} (.75,0);
			\draw [-|,drei] (.75,0) -- node [above] {$u_3$} (1.5,0);
			\draw [-|,dotted,drei] (1.5,0) -- 
				node [above,sloped] {$v_3$} +(\ang:.75);
			\draw [-|,dotted,drei] (1.5,0)++(\ang:.75) --
				node [above,sloped] {$y_3$} ++(\ang:.75);
			
			\draw [-|,zwei] (1.5,0) -- node [below,pos=.25] {$v_2$} (2,0);
			\draw [-|,zwei] (2,0) -- node [below] {$y_2$} (2.75,0);
			\draw [-|,dotted,zwei] (1.5,0) -- +(180+\ang:.25);
			\draw [-|,dotted,zwei] (1.5,0)++(180+\ang:.25) --
				node [below,sloped] {$u_2$} ++(180+\ang:.75);
			\draw [-|,dotted,zwei] (1.5,0)++(180+\ang:1) --
				node [below,sloped] {$x_2$} ++(180+\ang:.75);
		\end{scope}
	\end{tikzpicture}
	\caption{If no part of $x_1u_1v_1y_1$ is a factor of the splicing result,
		then the two splicings can be reduced to one splicing.}
	\label{fig:combine:rules}
\end{figure}
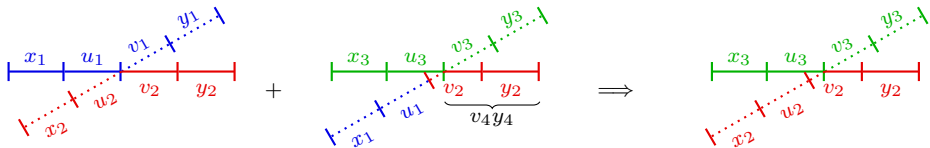

\begin{lemma}\label{lem:combine:rules}
	Let $L$ be a language,
	$w_i = x_i u_iv_i y_i \in L$ for $i = 1,2,3$,
	and $r_1 = (u_1,v_1;u_2,v_2)$, $r_2 = (u_3,v_3;u_4,v_4)$ be rules
	respecting $L$.
	If there are splicings
	\begin{align*}
		(w_1,w_2)&\splice{r_1} x_1u_1v_2y_2 = w_4 = x_4u_4v_4y_4, &
		(w_3,w_4)&\splice{r_2} x_3u_3v_4y_4 = z
	\end{align*}
	where $v_4y_4$ is a suffix of $v_2y_2$,
	then there is a rule $s = (u_3,v_3;u_2\del,\ti v_4)$
	which respects $L$ and $(w_3,w_2)\splice{s} z$.
	Furthermore, $\ti v_4 = v_4$ or $\ti v_4 \eLL v_2$.
\end{lemma}

\begin{proof}
Extension (Lemma~\ref{lem:extension}) justifies the assumption that the factors $u_1v_2$ and $u_4v_4$ match in $w_4$:
let $w_4\pos{i,j} = u_1v_2$ and $w_4\pos{i',j'} = u_4v_4$,
\begin{itemize}
	\item if $i < i'$ we extend $u_4$ in $r_2$ to the left by $i'-i$ letters,
	\item if $i > i'$ we extend $u_1$ in $r_1$ to the left by $i-i'$ letters,
	\item if $j < j'$ we extend $v_2$ in $r_1$ to the right by $j'-j$ letters, and
	\item if $j > j'$ we extend $v_4$ in $r_2$ to the right by $j-j'$ letters.
\end{itemize}
Clearly, the extended factors $u_1v_2$ and $u_4v_4$ match in $w_4$.
As $v_4y_4$ was a suffix of $v_2y_2$ before extension, now, $v_4$ is a suffix of $v_2$ and $y_2 = y_4$.
Additionally, either $v_4$ was not extended or $v_4 \eLL v_2$ and $v_2$ was not extended.
Let $\del$ such that $\del v_4 = v_2$, let $s = (u_3,v_3;u_2\del,v_4)$, and 	observe that $(w_3,w_2)\splice{s} z$.
	
Next, let us prove that $s$ respects $L$.
Let $w_3' = x_3' u_3v_3 y_3'\in L$ and 
$w_2' = x_2' u_2\del v_4 y_2' = x_2' u_2 v_2 y_2'\in L$.
There are splicings
\begin{align*}
	(w_1, w_2')&\splice{r_1} x_1u_1v_2 y_2' = w_4'
		= x_1u_4v_4 y_2', &
	(w_3', w_4')&\splice{r_2} x_3'u_3v_4y_2' = z'
\end{align*}
and $z'\in L$, concluding that $s$ respects $L$.
\end{proof}

Consider a splicing system $(J,S)$ and its generated language $L = L(J,S)$.
Let $n$ be the length of the longest word in $J$ and
let $\mu$ be the length-lexicographically largest
word that is a component of a rule in $S$.
Define $W_\mu = \set{w\in\Sig^*}{w\eLL \mu}$ as the set of words 
which are at most as large as $\mu$, in length-lexicographic order.
Furthermore, let $I = \Sig^{\le n}\cap L$ be a set of axioms and let
\[
	R = \set{r\in W_\mu^4}{r\text{ respects }L}
\]
be a set of rules.
It is not difficult to see that $J \sse I$, $S\sse R$, and $L = L(I,R)$.
Whenever convenient, we will assume that a splicing language $L$
is generated by a splicing system which is of the form of $(I,R)$.

Now, let us consider a word $xzy\in L$ where the length of the middle factor $z$ is at least $\abs \mu$.
The creation of $xzy$ by splicing in $(I,R)$ can be traced back to a word $x_1zy_1 = z_1$ where either $z_1\in I$ or where $z_1$ is created by a splicing that affects the factor $z$, \ie the splicing position lies in the factor $z$.
The next lemma describes this creation of $xzy = z_{k+1}$ by $k$ splicings in $(I,R)$, and shows that we can choose the rules and words which are used to create $z_{k+1}$ from $z_1$ such that the words are not significantly longer than $\ell = \max\oneset{\abs x,\abs y}$ and such that the rules satisfy certain length restrictions.

\begin{lemma}\label{lem:creation}
	Let $L$ be a splicing language, let $\ell,n\in\N$, let $m = \abs{M_L}$,
	and let $\mu$ be a word with $\abs\mu \ge \ell+2m$
	such that for $I = \Sig^{\le n}\cap L$ and
	$R = \set{r\in W_\mu^4}{r\text{ respects }L}$
	we have $L = L(I,R)$.

	Let $z_{k+1} = x_{k+1}zy_{k+1}$ with $\abs z \ge \abs\mu$ and
	$\abs{x_{k+1}},\abs{y_{k+1}}\le\ell$
	be a word that is created by $k$ splicings from a word $z_1 = x_1 z y_1$
	where either $z_1\in I$ or $z_1$ is created
	by a splicing $(w_0,w_0')\splice{s} z_1$ 
	where $w_0, w_0'\in L$, $s$ respects $L$,
	and the splicing position lies in the factor $z$.
	Furthermore, for $i = 1,\ldots,k$ the intermediate splicings are either
	\begin{enumerate}[(i)]
		\item $(w_i,z_i)\splice{r_i} x_{i+1}zy_{i+1}=z_{i+1}$,
			$w_i\in L$, $r_i\in R$, $y_{i+1} = y_i$,
			and the splicing position lies at the left of the factor $z$ or
		\item $(z_i,w_i)\splice{r_i} x_{i+1}zy_{i+1}=z_{i+1}$,
			$w_i\in L$, $r_i\in R$, $x_{i+1} = x_i$,
			and the splicing position lies at the right of the factor $z$.
	\end{enumerate}
	There are rules and words creating $z_{k+1}$, as above,
	satisfying in addition:
	\begin{enumerate}
		\item There is $k' \le k$ such that for $i = 1,\ldots,k'$
			all splicings are of the form \I and
			for $i = k'+1,\ldots,k$ all splicings are of the form \II.
		\item 
			For $i = 1,\ldots,k'$ the following bounds apply:
			$\abs{x_i} < \ell+2m$, $\abs{w_i} < \ell+2m$,
			$r_i \in \Sig^{<\ell+m}\times\Sig^{<2m}\times\Sig^{<2m}\times W_\mu$,
			and $x_{k'+1} = x_{k'+2} = \cdots = x_{k+1}$.
		\item
			For $i = k'+1,\ldots,k$ the following bounds apply:
			$\abs{y_i} < \ell+2m$, $\abs{w_i} < \ell+2m$,
			$r_i \in W_\mu\times\Sig^{<2m}\times\Sig^{<2m}\times 
			\Sig^{<\ell+m}$,
			and $y_{1} = y_{2} = \cdots = y_{k'+1}$.
	\end{enumerate}
	In particular, if $n \ge\ell+2m$, then
	$w_1,\ldots,w_k\in I$.
\end{lemma}

\begin{proof}	
	The first statement follows immediately by Remark~\ref{rem:swap}
	and the fact that $\abs{z}\ge \abs\mu$.
	The first statement also implies
	implies $x_{k'+1} = x_{k'+2} = \cdots = x_{k+1}$
	and $y_1 = y_2 = \cdots = y_{k'+1}$.
	Note that if $k' = 0$ (or $k' = k$), then statement~2
	(\resp statement~3) is trivially true.
	
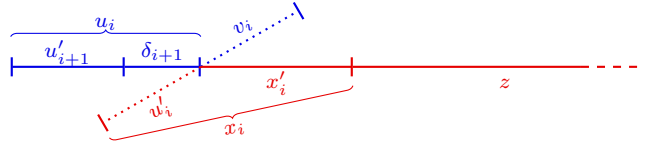
\begin{figure}[ht]
	\centering
	\begin{tikzpicture}[text height=1ex,text depth=0ex,
			font={\footnotesize}]
		\begin{scope}[thick]
			\def\ang{30}

			\draw [|-|,eins] (-.5,0) -- node [above] {$u_{i+1}'$} (1,0);
			\draw [-|,eins] (1,0) -- node [above] {$\del_{i+1}$} (2,0);
			\draw [-|,dotted,eins] (2,0) -- 
				node [above,sloped] {$v_i$} +(\ang:1.5);
				
			\draw [-|,zwei] (2,0) -- node [below] {$x_i'$} (4,0);
			\draw [-,zwei] (4,0) -- 
				node [below,pos=.6666666] {$z$} (7,0);
			\draw [-,zwei,dashed] (7,0) -- (8,0);
			\draw [-|,dotted,zwei] (2,0) -- 
				node [below,sloped] {$u_i'$} +(180+\ang:1.5);
		\end{scope}

		\draw [decorate,decoration=brace,zwei] (4,-.25) to 
			node [below,sloped] {$x_i$} (.8,-.95);
		\draw [decorate,decoration=brace,eins] (-.5,.4) to 
			node [above] {$u_i$} (2,.4);
	\end{tikzpicture}
	\caption{The $i$-th splicing for $i\le k'$ in the proof of Lemma~\ref{lem:creation} where $x_{i+1} = u_i x_i'$ and $v_i'$ is a prefix of $ x_i'z$.}
	\label{fig:series:proof:four}
\end{figure}

	The notation we employ in order to prove statement~2
	is chosen such that it matches with Figure~\ref{fig:series:proof:four}.
	For $i = 1,\ldots,k'$, let $r_i = (u_i,v_i;u_i',v_i')$.
	By extension (Lemma~\ref{lem:extension}),
	we may assume that $w_i = u_iv_i$
	and $x_i = u_i' x_i'$ such that
	$x_{i+1} = u_i x_i'$ and $ v_i'$ is a prefix of $x_i'z$.
	Let $ x_{k'+1}' = x_{k'+1} = x_{k+1}$ and $u_{k'+1}' = \e$.
	By Lemma~\ref{lem:combine:rules},
	we may assume that every splicing position
	lies at the left of the previous splicing position,
	\ie $x_i'$ is a proper suffix of $x_{i+1}'$ and
	$\abs{x_i'} \le \ell$ as $\abs{x_{k'+1}'} \le \ell$.
	Due to the modifications we made, we may have lost control of the lengths of 
	$u_i$, $v_i$, and $u_i'$; but $v_i'$ still belongs to $W_\mu$ and
	$r_i$ respects $L$.
	Let $\del_{i+1}$ such that $x_{i+1}' = \del_{i+1}x_i'$;
	hence, $u_i = u_{i+1}' \del_{i+1}$.
	The factor $\del_{i+1}$ is the the part of $x_{k+1}$ which is added by the $i$-th
	splicing and is not modified afterwards;
	$x_{k+1} = \del_{k'+1}\cdots\del_2x_1'$.
	Now, for $i = 2,\ldots, k'$,
	we replace $u_i'$ by a shortest word from $[u_i']_L$.
	(We also replace this prefix of $x_i$ and $u_{i-1}$.)
	Furthermore, we replace $v_i$ by a shortest word from $[v_i]_L$
	for $i = 1,\ldots,k'$.
	By Lemma~\ref{lem:smallest:words}, we have $\abs{u_i'}, \abs{v_i} < m$.
	We do not replace $u_1'$ yet, as this might affect 
	the word $w_0$ and the rule $s$
	in the splicing $(w_0,w_0')\splice s x_1zy_1$.
	
	Observe that the words $z_i$, $w_i$, and the rules $r_i$ 
	can still be used to create $z_{k+1}$ by splicing,
	in the way described in the claim.
	For $i=2,\ldots,k'$, we have $\abs{x_i} = \abs{u_i'x_i'} <\ell+m$,
	$\abs{w_i} \le \abs{x_{i+1}} +\abs{v_i} < \ell + 2m$, and
	$r_i \in \Sig^{<\ell+m}\times\Sig^{<m}\times\Sig^{<m}\times W_\mu$.
	We also have $\abs{w_1} < \ell+2m$ and
	$r_1 \in \Sig^{<\ell+m}\times\Sig^{<m}\times\Sig^*\times W_\mu$.	
	Note that, except for the length of $x_1$,
	and the third component of $r_1$,
	we have proven statement~2 (of the lemma) and
	we actually have proven a stronger bound than claimed.
	Symmetrically, we can consider statement~3 to be proven
	except for $y_1 = y_{k'+1}$
	and the second component of $r_{k'+1}$.

	Let $x_1 = u_1'x_1'$ as above and, symmetrically,
	let $y_1 = y_{k'+1}'v_{k'+1}'$ where $v_{k'+1}'$
	is the second component of $r_{k'+1}$.
	If $k' = 0$ (or $k'=k$), then $u_1'$ (\resp $v_{k'+1}'$)
	can be considered empty and $x_1' = x_{k+1}$
	(\resp $y_{k'+1}' = y_{k+1}$).
	If $z_1\in I$, we replace $u_1'$ and $v_{k'+1}'$ by
	shortest words from their syntactic classes, respectively,
	and the claim holds by Lemma~\ref{lem:smallest:words}.
	Otherwise, $(w_0,w_0') \splice s z_1$ where
	$u_1'$ is a prefix of $w_0$ and $v_{k'+1}'$ is a suffix of $w_0'$.
	
	Let $s = (u_0,v_0;u_0',v_0')$ and
	consider the overlap of the factor $u_0$ in the splicing
	$(w_0,w_0') \splice s z_1$ with the prefix $u_1'$ of $w_0$.
	In case when $u_0$ does not overlap with $u_1'$,
	replace $u_1'$ by a shortest word from its syntactic class.
	If $u_0$ and $u_1'$ overlap,
	let $u_1' = \del_1\del_2$ such that $\del_2$ is the overlap
	and replace $\del_1$ and $\del_2$ by shortest words from 
	their syntactic classes, respectively.
	Note that if we modified $u_1$, it got shorter;
	hence, $s$ still belongs to $R$.
	In any case, $\abs{u_1'} < 2m$, $\abs{x_1} < \ell+2m$
	(Lemma~\ref{lem:smallest:words}),
	and $r_1 \in \Sig^{<\ell+m}\times\Sig^{<m}\times\Sig^{<2m}\times W_\mu$;
	thus, the second statement.
	
	We may treat $v_{k'+1}'$ and $r_{k'+1}$ symmetrically
	in order to prove statement~3.
\end{proof}

\subsection{Proof of Theorem~\ref{thm:main:four}}\label{sec:proof:four}

Let $L$ be a splicing language and $m = \abs{M_L}$.
Throughout this section, by $\sim$ we denote the equivalence
relation $\sim_L$
and by $[\,\cdot\,]$ we denote the corresponding equivalence classes
$[\,\cdot\,]_L$.

Recall that Theorem~\ref{thm:main:four} claims that the
splicing system $(I,R)$ with $I = \Sig^{<m^2+ 6m}\cap L$ and
\[
	R = \set{r\in\Sig^{< m^2+10m}\times\Sig^{< 2m}\times
	\Sig^{< 2m}\times\Sig^{< m^2+10m}}
	{r\text{ respects }L}
\]
generates $L$.
The proof is divided in two parts.
In the first part, Lemma~\ref{lem:long:rules}, we prove that the set
of rules can be chosen as $\set{r\in (\Sig^{<m^2+10m})^4}{r\text{ respects }L}$
for some finite set of axioms.
The second part concludes the proof of Theorem~\ref{thm:main:four},
by employing the length bound $2m$ for the second and third component of rules
and by proving that the set of axioms can be chosen as
$I= \Sig^{<m^2+ 6m}\cap L$.

\begin{lemma}\label{lem:long:rules}
	Let $L$ and $m$ as above.
	There exists $n\in \N$ such that
	the splicing system $(I,R)$ with
	$I = \Sig^{\le n}\cap L$ and
	\[
		R = \set{r\in (\Sig^{<m^2+10m})^4}{r\text{ respects }L}
	\]
	generates the same language $L = L(I,R)$.
\end{lemma}

\begin{proof}
	As every word in $I$ belongs to $L$ and every rule in $R$
	respects $L$, the inclusion $L(I,R) \sse L$ holds
	(for any $n$).
	
	Since $L$ is a splicing language, there exists a splicing system $(I',R')$
	which generates $L$.
	Let $n'$ be a number larger	than  any word in $I'$ and larger
	than any component of a rule in $R'$ and let $n = n' +6m$.
	Let $I = \Sig^{\le n}\cap L$ as in the claim and observe that $L(I,R') = L$.
	
	For a word $\mu$ we let $W_\mu = \set{w\in\Sig^*}{w\eLL \mu}$,
	as we did before.
	Define the set of rules where every component
	is length-lexicographically bounded by $\mu$
	\begin{align*}
		R_\mu &= \set{r \in W_\mu^4}{r\text{ respects } L}
	\end{align*}
	and the language $L_\mu = L(I,R_\mu)$; clearly, $L_\mu \sse L$.
	For two words $\mu \eLL v$ we see that $R_\mu\sse R_v$,
	and hence, $L_\mu\sse L_v$.	
	Thus, if $L_\mu = L$ for some word $\mu$, then
	for all words $v$ with $\mu \eLL v$, we have $L_v = L$.
	As $L = L(I,R')$,
	there exists a word $\mu$ such that $L_\mu = L$ and $\abs\mu+6m \le n$.
	Let $\mu$ be the smallest word, in the length-lexicographic order, such that
	$L_\mu = L$.
	Note that if $\abs\mu < m^2+10$, then $R_\mu \sse R$ and
	$L = L_\mu \sse L(I,R)$.
	For the sake of contradiction assume $\abs \mu \ge m^2+10m$.
	Let $\nu$ be the next-smaller word than $\mu$,
	in the length-lexicographic order,
	and let $S = R_{\nu}$.
	Note that $L(I,S) \ssne L$ and $R_\mu\sm S$ contains only rules
	which have a component that is equal to $\mu$.

	Choose $w$ from $L\sm L(I,S)$ as a shortest word, \ie
	for all $w'\in L$ with $\abs {w'} < \abs w$,
	we have $ w' \in L(I,S)$.
	Factorize $w = xzy$ with $\abs x = \abs y = 3m$
	and note that $\abs z \ge \abs{\mu}$, otherwise $w\in I$.
	Factorize $\mu = \del_1\abg\del_2$ with
	$\abs{\del_1},\abs{\del_2} \ge 5m$,
	$\abs{\abg} = m^2$, $\bet\neq \e$, $\alp\sim \alp\bet$,
	and $\gam \sim \bet\gam$ (Lemma~\ref{lem:pumping}).
	
	We will show that there is a series of splicings which creates
	$w$ from a set of shorter words and by using splicing rules from $S$.
	This yields a contradiction to the choice of $w$.
	In order to find this series of splicings we investigate the creation
	of a word $x\ti z y$ where $\ti z$ is derived by using a pumping argument on
	all factors $\abg$ in $z$.
	
	Let $j$ be a sufficiently large even number 
	($j > 4 \abs \mu+\abs z$ will suffice).
	We define a word $\ti z$ which is the result of applying the
	pumping algorithm from Lemma~\ref{lem:pumping:algorithm} on $z$,
	as discussed in Section~\ref{sec:pump:alg}.
	The pumping algorithm replaces the occurrences of $\abg$ in $z$ by $\abjg$
	such that for every factor $\ti z\pos{k,k+m^2} = \abg$, either
	\begin{enumerate}[(a)]
		\item $\alpha\beta^{j\slash 2}$ is a factor of $\ti z$ starting at position
			$\ti z\pos{k}$ or 
		\item $\beta^{j\slash2}\gamma$ is a factor of
			$\ti z$ ending at position $\ti z\pos{k+m^2}$
	\end{enumerate}
	holds.
	In particular, if $\del_1\abg\del_2$ is a factor of $\ti z$ either
	\A $\gamma\del_2$ is a prefix of a word in $\bet^+$ or
	\B $\delta_1\alpha$ is a suffix of a word in $\bet^+$.
	By induction and as $\abg \sim \abjg$, it is easy to see that
	$z\sim \ti z$ and $x\ti zy\in L$.
	
	Let us trace back the creation of $x\ti zy\in L$
	by splicing in $(I,R_\mu)$ to a word $x_1\ti zy_1$
	where either $x_1\ti zy_1\in I$
	or where $x_1\ti zy_1$ is created by a splicing
	that affects $\ti z$, \ie the splicing position lies within the factor $\ti z$.
	Let $z_{k+1} = x_{k+1}\ti zy_{k+1}$, where $x_{k+1} = x$ and $y_{k+1} = y$,
	be created by $k$ splicings from a word $z_1 = x_1\ti z y_1$
	where either $x_1\ti z y_1\in I$ or $x_1\ti zy_1$ is created
	by a splicing $(w_0,w_0')\splice{s} z_1$ with $w_0,w_0'\in L$,
	$s\in R_\mu$, and the splicing position lies in the factor $\ti z$.
	Furthermore, for $i = 1,\ldots,k$ the intermediate splicings are either
	\begin{enumerate}[(i)]
		\item $(w_i,z_i)\splice{r_i} x_{i+1}\ti zy_{i+1}=z_{i+1}$,
			$w_i\in L$, $r_i\in R_\mu$, $y_{i+1} = y_i$,
			and the splicing position lies at the left of the factor $\ti z$ or
		\item $(z_i,w_i)\splice{r_i} x_{i+1}\ti zy_{i+1}=z_{i+1}$,
			$w_i\in L$, $r_i\in R_\mu$, $x_{i+1} = x_i$,
			and the splicing position lies at the right of the factor $\ti z$.
	\end{enumerate}
	Note that $\abs{\ti z} \ge \abs z \ge \abs \mu$ and,
	therefore, we can apply Lemma~\ref{lem:creation} (with $\ell = 3m$).
	Thus, we may assume that $w_i \in I$
	and $\abs{x_i},\abs{y_i} < 5m$ for $i = 1,\ldots,k$.
	
	Consider a rule $r_i$ in a splicing of the form \I.
	By Lemma~\ref{lem:creation},
	$r_i\in \Sig^{<4m}\times\Sig^{<2m}\times\Sig^{<2m}\times W_\mu$.
	Suppose the fourth component of $r_i$ covers a prefix of
	the left-most factor $\alp\bet^{j/2}$ in $\ti z$
	which is longer than $\alp$
	(as $j$ is very large, it cannot fully cover $\alp\bet^{j/2}$).
	By extension (Lemma~\ref{lem:extension}),
	we may write $r_i = (u_1,v_1;u_2,v'\alp\bet^h)$
	for some $h\ge 1$.
	By Lemma~\ref{lem:congruence} and as $\alp \sim \alp\bet$,
	we may replace this rule by $(u_1,v_1;u_2,v'\alp)$.
	Note that, as the fourth component got shorter, now $r_i\in S$.
	
	After we symmetrically treated rules of form \II,
	these new rules $r_1,\ldots,r_k$ and the words $w_1,\ldots,w_k$
	can be used in order to create
	$w = x_{k+1}zy_{k+1}$ from $x_1zy_1$ by splicing.
	In order to see this, observe that,
	even though the factors $\abg$ in $z$, which
	we pumped up before, may overlap with each other,
	the left-most (and right-most) position where we replaced $\bet$
	by $\bet^j$ is preceded by the factor $\alp$ 
	(\resp succeeded by the factor $\gam$) in $\ti z$.

	Next, we show that all the rules $r_1,\ldots,r_k$ belong to $S$, now.
	By contradiction, suppose $r_i\notin S$ for some $i$ and, by symmetry,
	suppose this $i$-th splicing is of the form \I.
	Thus, the fourth component of $r_i$ has to be $\mu = \del_1\abg\del_2$.
	As $\abs{\del_1} \ge 5m > \abs{x_i}$,
	the factor $\abg$ in $\mu$ is covered by $\ti z$.
	Let $k$ such that $\abg = \ti z\pos{k;k+m^2}$ is this factor in $\ti z$.
	The pumping algorithm ensured that
	\A $\alp\bet^{j\slash2}$ is a factor of $\ti z$ starting at position
	$\ti z\pos k$ or 
	\B $\bet^{j\slash2}\gam$ is a factor of $\ti z$ ending at position
	$\ti z\pos{k+m^2}$.
	As $j/2$ is very large and the splicing position of
	$(w_i,z_i)\splice{r_i} z_{i+1}$ is too close
	to the left end of $z_{i+1}$, case \B is not possible.
	Thus, case \A holds, the fourth component of $r_i$ overlaps
	in more than $\abs{\alp}$ letters with the left-most factor $\alp\bet^{j/2}$
	in $\ti z$, and we used the replacement above which ensured $r_i\in S$
	--- contradiction.

	Let us summarize:
	if $x_1zy_1$ was in $L(I,S)$, then $w\in L(I,S)$ as well,
	which would contradict the choice of $w$.
	If $z_1 = x_1\ti z y_1 \in I$, then $x_1 zy_1$, which is
	at most as long as $z_1$, would belong to $I$ and we are done.
	We only have to consider the case when
	$(w_0,w_0')\splice{s} z_1 = x_1\ti z y_1$
	and the splicing position lies within the factor $\ti z$.
	We will show that, from this splicing, we derive another splicing
	$(\ti w_0,\ti w_0')\splice t x_1zy_1$ which respects $L(I,S)$
	and, therefore, yields the contradiction.
	
	Let $s = (u,v_1;u_2,v)$,
	$w_0 = x uv_1$ and $w_0' = u_2v y$
	where $\abs{v_1},\abs{u_2} < m$,
	by Lemma~\ref{lem:shorten}
	(here, $x$ and $y$ are newly chosen words).
	We have
	\[
		z_1 = x_1\ti zy_1 = xuvy
	\]
	where $xu$ is a proper prefix of $x_1\ti z$
	and $vy$ is a proper suffix of $\ti z y_1$.

	We will see next that if $s\notin S$,
	then we can use a rule $\ti s\in S$ and 
	maybe slightly modified words
	in order to obtain $z_1$ by splicing.
	If $s\notin S$, then $u = \mu$ or $v = \mu$.
	Suppose $u = \mu = \del_1\abg\del_2$.
	Thus, the factor $\abg$ of $\mu$ is covered by the factor $\ti z$ in $z_1$ as 
	$\abs{\del_1} \ge 5m > \abs{x_1}$.
	Let $\abg = \ti z\pos{k;k+m^2}$ be this factor.
	\A $\alp\bet^{j\slash2}$ is a factor of $\ti z$ starting at position
	$\ti z\pos k$ or 
	\B $\bet^{j\slash2}\gam$ is a factor of $\ti z$ ending at position
	$\ti z\pos{k+m^2}$.
	If \B holds, $\del_1\alp$ is a suffix of a word in $\bet^+$.
	We may write $\del_1\alp = \bet_2\bet^\ell$ where $\ell \ge 0$
	and $\bet_2$ is a suffix of $\bet$.
	Replace $u$ by $\bet_2\gam\del_1$
	and use this new rule $\ti s$ in order to splice
	$(w_0,w_0')\splice{\ti s} z_1$.
	Note that the first component is now shorter than $\mu$.	
	Otherwise, \A holds and $\gam\del_2 v$ is a prefix of
	a word in  $\bet^+$.	
	As $j$ is very large and $\gam$ is a prefix of a word in $\bet^+$,
	we may extend $v$ (Lemma~\ref{lem:extension})
	such that we can write $\bet\gam\del_2 = \bet^{\ell_1}\bet_1$
	and $v = \bet_2\bet^{\ell_2}\gam$
	where $\ell_1\ge 1$, $\ell_2 \ge 0$ and $\bet_1\bet_2 = \bet$.
	Now, we pump down one of the $\bet$ in the first component
	and $\bet^{\ell_2}$ in the fourth component and
	we let $\ti s = (\del_1\alp\bet^{\ell_1-1}\bet_1,v_1;u_2,\bet_2\gam)\sim s$.
	As all components are shorter than $\mu$, we see that $\ti s \in S$ and
	\[
		(x\del_1\alp\bet^{\ell_1-1}\bet_1v_1,u_2\bet_2\bet^{\ell_2+1}\gam y)
			\splice{\ti s} z_1,
	\]
	\ie we have shifted one of the occurrences of $\bet$ from $w_0$ to $w_0'$.
	Note that $\bet_2\gam$ is a prefix of $\bet_2\bet^{\ell_2+1}\gam$.
	Treating the fourth component analogously
	justifies the assumption that $s \in S$.

	Next, we will pump down the factors $\alp\bet^j\gam$ to $\abg$
	in $\ti z$ again.
	At every position where we pumped up before, we are now pumping down
	(in reverse order)
	in order to obtain the words $\ti x,\ti u,\ti v,\ti y$ from
	the words $x,u,v,y$, respectively.
	For each pumping step do:
	\begin{asparaenum}
		\item If $u$ is covered by the factor $\abjg$
	(which we pump down in this step),
	extend $u$ to the left such that it becomes a prefix of $\abjg$.
	Symmetrically, if $v$ is covered by the factor $\abjg$,
	extend $v$ to the right such that it becomes a suffix of $\abjg$
	(Lemma~\ref{lem:extension}).
	Observe that extension ensures that the factor $\abjg$ is covered
	by either $xu$, $uv$, or $vy$.

	\item If $\alp\bet^j$ or $\bet^j\gam$ is covered by one of
	$x$, $u$, $v$, or $y$, then replace this factor
	by $\alp\bet$ or $\bet\gam$, respectively.
	Otherwise, let us show how to pump when
	$\alp\bet^j\gam$ is covered by $xu$.
	The cases when $\abjg$ is covered by $uv$ or $vy$ can be treated
	analogously.
	We can factorize $x = x'\alp\bet^{j_1}\bet_1$ and
	$u = \bet_2\bet^{j_2}\gam u'$
	where $\bet_1\bet_2 = \bet$ and $j_1+j_2+1 = j$.
	The pumping results are the words $x'\alp\bet_1$
	and $\bet_2\gam  u'$, respectively.
	\end{asparaenum}
	
	Observe that, after reversing all pumping steps,
	$\ti x\ti u \sim xu$, $\ti v \ti y\sim vy$,
	$\ti x\ti u\ti v\ti y = x_1zy_1$,
	and the rule $t = (\ti u, v_1;u_2,\ti v)$ respects $L$.
	Furthermore, if we used extension for $u$ (or $v$) in one of the steps,
	then $\abs{\ti u} \le m^2$ (\resp $\abs{\ti v}\le m^2$);
	in any case $t\in S$.
	Recall that $w$ was chosen as the shortest word from $L\sm L(I,S)$.
	As $\abs{\ti x\ti u v_1},\abs{u_2\ti v\ti y} < \abs{z}+6m =\abs w$,
	the words $\ti w_0 = \ti x\ti u v_1$ and
	$\ti w_0' = u_2\ti v\ti y$ belong to $L(I,S)$,
	and as $(\ti w_0,\ti w_0')\splice{t} x_1 z y_1$,
	we conclude that $x_1 z y_1$ as well as $w$ belong to $L(I,S)$ ---
	the desired contradiction.
\end{proof}

Now, we can prove our main result.

\begin{proof}[Proof of Theorem~\ref{thm:main:four}]
	Recall that for a splicing language $L$ with $m = \abs{M_L}$
	we intend to prove that the
	splicing system $(I,R)$ with $I = \Sig^{<m^2+ 6m}\cap L$ and
	\[
		R = \set{r\in\Sig^{< m^2+10m}\times\Sig^{< 2m}\times
		\Sig^{< 2m}\times\Sig^{< m^2+10m}}
		{r\text{ respects }L}
	\]
	generates the language $L = L(I,R)$.
	
	Obviously, $L(I,R) \sse L$.
	By Lemma~\ref{lem:long:rules},
	we may assume that $L$ is generated by a splicing system
	$(J,S)$ where
	\[
		S = \set{r\in(\Sig^{< m^2+10m})^4}
		{r\text{ respects }L}.
	\]

	In order to prove $L\sse L(I,R)$,
	we use induction on the length of words in $L$.
	For $w\in L$ with $\abs w < m^2+6m$,
	by definition, $w\in I \sse L(I,R)$.

	Now, consider $w\in L$ with $\abs w \ge m^2+6m$.
	The induction hypothesis states that
	every word $w'\in L$ with $\abs{w'} < \abs {w}$ belongs to $L(I,R)$.
	Factorize $w = x\alp\bet\gam\del y$ such that
	$\abs{x} = \abs{y} = 3m$,
	$\abs{\alp\bet\gam} = m^2$, $\bet \ne \e$,
	$\alp \sim \alp\bet$, and $\gam \sim \bet\gam$ (Lemma~\ref{lem:pumping}).

	The proof idea is similar as in the proof of Lemma~\ref{lem:long:rules}.
	We use a pumping argument on $\bet$ in order to obtain
	a very long word.
	This word has to be created by a series of splicings in $(J,S)$.
	We show that these splicings can be
	modified in order to create $w$ by splicing from
	a set of strictly shorter words and with rules from $R$.
	Then, the induction hypothesis yields $w\in L(I,R)$.

	Choose $j$ sufficiently large ($j > \abs{w}+m^2+10m$ and 
	$J$ does not contain words of length $j$ or more).
	We let $z = \abjgd$ and
	investigate the creation of $xzy \in L$ by splicing in $(J,S)$.
	As $z$ is not a factor of a word in $J$, we can 
	trace back the creation of $xzy$ by splicing to the point
	where the factor $z$ is affected for the last time.
	Let $z_{k+1} = x_{k+1}zy_{k+1}$, where $x_{k+1} = x$ and $y_{k+1} = y$,
	be created by $k$ splicings from a word $z_1 = x_1z y_1$
	which is created
	by a splicing $(w_0,w_0')\splice{s} z_1$ with $w_0, w_0'\in L$,
	$s\in S$, and the splicing position lies in the factor $z$.
	Furthermore, for $i = 1,\ldots,k$ the intermediate splicings are either
	\begin{enumerate}[(i)]
		\item $(w_i,z_i)\splice{r_i} x_{i+1}zy_{i+1}=z_{i+1}$,
			$w_i\in L$, $r_i\in S$, $y_{i+1} = y_i$,
			and the splicing position lies at the left of the factor $z$ or
		\item $(z_i,w_i)\splice{r_i} x_{i+1}zy_{i+1}=z_{i+1}$,
			$w_i\in L$, $r_i\in S$, $x_{i+1} = x_i$,
			and the splicing position lies at the right of the factor $z$.
	\end{enumerate}
	As $\abs{z} \ge m^2+10m$ we can apply Lemma~\ref{lem:creation}.
	Thus, we may assume $w_1,\ldots,w_k \in I$, $r_1,\ldots,r_k\in R$,
	and $\abs{x_1},\abs{y_1} < 5m$.

	Consider a rule $r_i$ in a splicing of the form \I.
	Suppose the fourth component of $r_i$ covers a prefix of
	the factor $\alp\bet^j$ in $z$ which is longer than $\alp\bet$
	(as $j$ is very large, it cannot fully cover $\alp\bet^{j}$).
	By extension (Lemma~\ref{lem:extension}),
	we may write $r_i = (u_1,v_1;u_2,v'\alp\bet^\ell)$
	for some $\ell\ge 1$.
	By Lemma~\ref{lem:congruence} and as $\alp \sim \alp\bet$,
	we may replace this rule by $(u_1,v_1;u_2,v'\alp) \in R$.
	Moreover, after we symmetrically treated rules of form \II,
	these new rules $r_1,\ldots,r_k$ and the words $w_1,\ldots,w_k$
	can be used in order to create
	$w = x_{k+1}\abgd y_{k+1}$ from $x_1\abgd y_1$ by splicing.
	Thus, if $x_1\abgd y_1$ belongs to $L(I,R)$, so does $w$.

	Now, consider the first splicing
	$(w_0, w_0')\splice s z_1 = x_1zy_1$.
	By Lemma~\ref{lem:shorten},
	let $s = (u,v_1;u_2,v)$ such that
	$w_0 = xuv_1$, $w_0' = u_2v y$
	and $\abs{v_1},\abs{u_2} <m$
	(here, $x$ and $y$ are newly chosen words).
	Hence,
	\[
		z_1 = xuvy = x_1z y_1 = x_1\abjgd y_1
	\]
	where $xu$ is a proper prefix of $x_1z$ and $vy$ is a proper suffix of $z y_1$.

	Next, we will pump down the factor $\alp\bet^j\gam$ to $\abg$
	in $ z$ again in order to obtain the words $\ti x,\ti u,\ti v,\ti y$
	from the word $x,u,v,y$, respectively.
	The pumping is done as in the proof of Lemma~\ref{lem:long:rules}:

	\begin{asparaenum}
	\item If $u$ is covered by the factor $\abjg$,
	extend $u$ to the left such that it becomes a prefix of $\abjg$.
	Symmetrically, if $v$ is covered by the factor $\abjg$,
	extend $v$ to the right such that it becomes a suffix of $\abjg$
	(Lemma~\ref{lem:extension}).
	Observe that extension ensures that the factor $\abjg$ is covered
	by either $xu$, $uv$, or $vy$.
	\item If $\alp\bet^j$ or $\bet^j\gam$ is covered by one of
	$x$, $u$, $v$, or $y$, then replace this factor
	by $\alp\bet$ or $\bet\gam$, respectively.
	Otherwise, let us show how to pump when
	$\alp\bet^j\gam$ is covered by $xu$.
	The cases when $\abjg$ is covered by $uv$ or $vy$ can be treated
	analogously.
	We can factorize $x = x'\alp\bet^{j_1}\bet_1$ and
	$u = \bet_2\bet^{j_2}\gam u'$
	where $\bet_1\bet_2 = \bet$ and $j_1+j_2+1 = j$.
	The pumping result are the words $x'\alp\bet_1$
	and $\bet_2\gam  u'$, respectively.
	\end{asparaenum}

	Observe that,
	$\ti x\ti u \sim xu$, $\ti v \ti y\sim vy$,
	$\ti x\ti u\ti v\ti y = x_1\abgd y_1$,
	and the rule $t = (\ti u, v_1;u_2,\ti v)$ respects $L$.
	Furthermore, if we used extension for $u$ (or $v$),
	then $\abs{\ti u} \le m^2$ (\resp $\abs{\ti v}\le m^2$).
	No matter whether we used extension, $t\in R$.
	As $\abs{\ti x\ti u v_1},\abs{u_2\ti v\ti y} < \abs{z}+6m =\abs w$
	and by induction hypothesis,
	the words $\ti w_0 = \ti x\ti u v_1$ and $\ti w_0 = u_2\ti v\ti y$ belong 
	to $L(I,S)$.
	We conclude that $(\ti w_0,\ti w_0')\splice t x_1\abgd y_1\in L(I,R)$
	and, therefore, $w=x_{k+1}\abgd y_{k+1}\in L(I,R)$ as well.
\end{proof}

\section{Decidability}\label{sec:decidability}

The main question we intended to answer when starting our investigation was, whether or not it is decidable if a given regular language $L$ is a splicing language.
If we can decide whether a splicing rule respects a regular language
and if we can construct a (non-deterministic) finite automaton accepting
the language generated by a given splicing system,
then we can decide whether $L$ is a classic splicing language
(Pixton splicing language) as follows.
We compute the splicing system $(I,R)$ as given in Theorem~\ref{thm:main:four}
(\resp Theorem~\ref{thm:main}) and we compute a finite automaton accepting
the splicing language $L(I,R)$.
Theorem~\ref{thm:main:four} (\resp Theorem~\ref{thm:main})
implies that $L$ is a splicing language if and only if $L = L(I,R)$.
Recall that equivalence of regular languages is decidable, \eg by constructing
and comparing the minimal deterministic finite automata of both languages.

It is known from \cite{HeadPG02,GoodePHD} that it is decidable
whether a classic splicing rule respects a regular language.
Furthermore, there is an effective construction of a
finite automaton which accepts the language generated by a Pixton splicing system
\cite{Pixton96}.
As mentioned earlier, Pixton splicing systems are more general than classic
splicing systems, which means the latter result applies to classic splicing systems, too.
Such a construction for classic splicing systems is also given in \cite{HeadP06}.

Let us prove that it is decidable whether a Pixton splicing rule $r$ respects
a regular language $L$.
Actually, we will decide whether the set $[r]_L$ respects $L$,
which is equivalent by Lemma~\ref{lem:congruence}.
The proof can easily be adapted in order to prove that it is decidable whether
a classic splicing rule respects $L$.

\begin{lemma}
	Let $L$ be a regular language and let $r$ be a Pixton splicing rule.
	It is decidable whether $r$ respects $L$.
\end{lemma}

\begin{proof}
	Let $\sim$ denote the equivalence relation $\sim_L$
	and $[\,\cdot\,]$ denote the corresponding equivalence classes $[\,\cdot\,]_L$.

	Let $r = (u_1,u_2;v)$.
	We define the two sets $S_1,S_2 \sse M_L$ as
	\begin{align*}
		S_1 &= \set{X\in M_L}{\exists Y \colon X[u_1]Y \sse L}, &
		S_2 &= \set{Y\in M_L}{\exists X \colon X[u_2]Y \sse L},
	\end{align*}
	\ie $[x_1]$ belongs to $S_1$ if and only if $x_1u_1y_1\in L$ for some
	word $y_1$ and $[y_2]$ belongs to $S_2$ if and only if $x_2u_2y_2\in L$ for some
	word $x_2$.
	We claim that $r$ respects $L$ if and only if $X[v]Y\sse L$ for all
	$X\in S_1$ and $Y\in S_2$, which is a property that can easily be decided.
	
	Firstly, suppose $r$ respects $L$.
	For $X\in S_1$ and $Y\in S_2$ choose words $x_1\in X$ and $y_2\in Y$.
	By definition of $S_1$ and $S_2$, there is $y_1$ and $x_2$ such that
	$x_iu_iy_i \in L$ for $i = 1,2$ and,
	as $r$ respects $L$, $x_1vy_2\in L$.
	This implies $X[v]Y\sse L$.
	
	Vice verse, suppose $X[v]Y \sse L$ for all $X\in S_1$ and $Y\in S_2$.
	For all $x_iu_iy_i\in L$ with $i = 1,2$,
	we have $[x_1]\in S_1$ and $[y_2]\in S_2$.
	Therefore, $x_1vy_2 \in [x_1][v][y_2] \sse L$ and $r$ respects $L$.
\end{proof}

These observations lead to the decidability results.

\begin{corollary}\label{cor:decidability}\nopagebreak
\ \\ \vspace{-\baselineskip}\nopagebreak
\begin{enumerate}[i.)]	
	\item	
	For a given regular language $L$, it is decidable whether or not $L$
	is a classic splicing language.
	Moreover, if $L$ is a classic splicing language,
	a splicing system $(I,R)$ generating $L$ can be effectively constructed.
	\item	
	For a given regular language $L$, it is decidable whether or not $L$
	is a Pixton splicing language.
	Moreover, if $L$ is a Pixton splicing language,
	a splicing system $(I,R)$ generating $L$ can be effectively constructed.
\end{enumerate}
\end{corollary}

\section*{Final Remarks}

It has been known since 1991 that the class \cS of languages that can be generated by a splicing system is a proper subclass of the class of regular languages.
However, to date, no other natural characterization for the class \cS exists.
The problem of deciding whether a regular language is generated by a splicing system is a fundamental problem in this context and has remained unsolved.
To the best of our knowledge, the problem was first stated in the literature in 1998 \cite{Head98}.
In this paper we solved this long standing open problem.

Regarding the complexity of the decision algorithm, let $L$ be a regular language given as syntactic monoid $M_L$ and $(I,R)$ be the splicing system described in Theorem~\ref{thm:main:four} (\resp Theorem~\ref{thm:main}).
An automaton which accepts $L(I,R)$ and is created as described in Section~\ref{sec:decidability} has a state set of size in $2^{\Oh(m^2)}$, where $m = \abs{M_L}$.
Deciding the equivalence of two regular languages, given as NFAs, is known to be \PSPACE-complete \cite{StockmeyerM73};
hence, the naive approach to decide whether or not $L = L(I,R)$ uses double exponential time $2^{2^{\Oh(m^2)}}$.
As there may be an exponential gap between an NFA accepting $L$ and the syntactic monoid $M_L$, the complexity, when considering an NFA as input, becomes triple exponential.
Improving the complexity of the algorithm is subject of future research.


\end{document}